\documentclass[runningheads]{llncs}
\usepackage[utf8]{inputenc}
\usepackage[T1]{fontenc}
\usepackage{graphicx}
\usepackage{grffile}
\usepackage{longtable}
\usepackage{wrapfig}
\usepackage{rotating}
\usepackage[normalem]{ulem}
\usepackage{amsmath}
\usepackage{textcomp}
\usepackage{amssymb}
\usepackage{capt-of}
\usepackage{hyperref}
\usepackage[title]{appendix}
\usepackage{amstext,amsmath,amssymb,amsfonts}
\usepackage{colortbl}
\usepackage{xspace,varioref}
\usepackage{hyperref}

\usepackage[dvipsnames]{xcolor}
\usepackage{graphicx}

\usepackage{multicol}
\usepackage{algorithm} % Pour les pseudos-codes
\usepackage[noend]{algpseudocode}

\usepackage{enumitem}
\usepackage{stmaryrd}

\usepackage[most]{tcolorbox}
\newtcolorbox{mybox}[2][]{%
  attach boxed title to top left
               = {yshift=-7pt, xshift=15pt},
  colback      = white!5!white,
  colframe     = black!75!black,
  fonttitle    = \bfseries,
  colbacktitle = gray!85!black,
  title        = #2,#1,
  enhanced,
}
\newcommand{\latinloc}[1]{\ifx\undefined\lncs\relax\emph{#1}\else\textrm{#1}\fi\xspace}

\newcommand{\eg}{\latinloc{e.g.}}
\newcommand{\ie}{\latinloc{i.e.}}
\newcommand{\etal}{\textit{et al}. }

\makeatletter
\renewcommand*{\ALG@name}{Algorithm}
\makeatother

\DeclareMathOperator{\DB}{\mathbf{DB}}
\DeclareMathOperator{\PrivDB}{\mathbf{PrivDB}}
\DeclareMathOperator{\MultDB}{\mathbf{MultDB}}
\DeclareMathOperator{\PrivMultDB}{\mathbf{PrivMultDB}}
\DeclareMathOperator{\ByzPrivMultDB}{\mathbf{ByzPrivMultDB}}
\DeclareMathOperator{\ByzMultDB}{\mathbf{ByzMultDB}}
\DeclareMathOperator{\SMR}{\mathbf{SMR}}
\DeclareMathOperator{\ISMR}{\mathbf{ISMR}}
\DeclareMathOperator{\USMR}{\mathbf{USMR}}
\DeclareMathOperator{\cUSMR}{\mathbf{cUSMR}}

\DeclareMathOperator{\UpdKey}{\mathbf{UpdKey}}

\DeclareMathOperator{\Enc}{\mathsf{Enc}}
\DeclareMathOperator{\Dec}{\mathsf{Dec}}
\DeclareMathOperator{\KG}{\mathsf{KG}}
\DeclareMathOperator{\TG}{\mathsf{TG}}
\DeclareMathOperator{\Upd}{\mathsf{Upd}}

\DeclareMathOperator{\UE}{\mathsf{UE}}
\DeclareMathOperator{\enc}{\mathsf{enc}}
\DeclareMathOperator{\upd}{\mathsf{upd}}

\DeclareMathOperator{\MS}{\mathcal{M}}
\DeclareMathOperator{\CS}{\mathcal{C}}
\DeclareMathOperator{\KS}{\mathcal{K}}
\DeclareMathOperator{\TS}{\mathcal{T}}

\newcommand{\I}{\mathbf{I}}
\newcommand{\R}{\mathbf{R}}
\newcommand{\D}{\mathbf{D}}
\renewcommand{\H}{\mathbf{H}}
\newcommand{\CC}{\mathbf{C}}
\newcommand{\G}{\mathbf{G}}
\renewcommand{\SS}{\mathbf{S}}
\newcommand{\A}{\mathbf{A}}

\newcommand{\C}{\mathsf{C}}

\newcommand{\W}{\mathsf{W}}
\renewcommand{\S}{\mathsf{S}}
\newcommand{\uec}{\mathsf{ue}_{\mathsf{cli}}}
\newcommand{\ues}{\mathsf{ue}_{\mathsf{ser}}}
\newcommand{\honSrv}{\mathsf{honSrv}}

\newcommand{\K}{\mathcal{K}}
\newcommand{\DD}{\mathcal{D}}
\newcommand{\T}{\mathcal{T}}

\newcommand{\EE}{\mathcal{E}}

\newcommand{\IR}{\mathcal{IR}}
\newcommand{\FW}{\mathcal{FW}}
\renewcommand{\AA}{\mathcal{A}}
\newcommand{\RR}{\mathcal{R}}
\newcommand{\Q}{\mathcal{Q}}
\newcommand{\N}{\mathbb{N}}

\newcommand{\KK}{\mathbb{K}}
\newcommand{\TT}{\mathbb{T}}

\newcommand{\leaked}{\mathtt{leaked}}
\newcommand{\epoch}{\mathtt{epoch}}

\newcommand{\Key}{\mathsf{Key}}
\newcommand{\Token}{\mathsf{Token}}
\newcommand{\Data}{\mathsf{Data}}
\newcommand{\insec}{\mathtt{insec}}
\newcommand{\ENCCPA}{{\mathsf{ENC}\text{-}\mathsf{CPA}}}
\newcommand{\UPDCPA}{{\mathsf{UPD}\text{-}\mathsf{CPA}}}
\newcommand{\UECPA}{{\mathsf{IND}\text{-}\mathsf{UE}\text{-}\mathsf{CPA}}}

\newcommand{\CPA}{\mathsf{CPA}}
\newcommand{\CPAp}{\mathsf{CPA^+}}
\newcommand{\CCA}{\mathsf{CCA}}

\newcommand{\UPD}{\mathsf{UPD}}
\newcommand{\ENC}{\mathsf{ENC}}
\newcommand{\fwl}{\mathsf{fwl}}
\newcommand{\fwr}{\mathsf{fwr}}
\newcommand{\M}{\mathbb{M}}
\newcommand{\Hist}{\textsc{Hist}}
\newcommand{\Init}{\textsc{Init}}
\newcommand{\Active}{\textsc{Active}}

\allowdisplaybreaks

\date{\today}
\title{Interactivity in Constructive Cryptography : Modeling and Applications to Updatable Encryption and Private Information Retrieval}
\author{Françoise Levy-dit-Vehel\inst{1} \and Maxime Roméas\inst{2}}
\authorrunning{F. Levy-dit-Vehel and M. Roméas}

\institute{LIX, ENSTA Paris, INRIA, Institut Polytechnique de Paris, 91120 Palaiseau, France.
  \email{levy@ensta.fr} \and
  LIX, \'Ecole polytechnique, INRIA, Institut Polytechnique de Paris, 91120 Palaiseau, France.
\email{romeas@lix.polytechnique.fr}}

\begin{document}

\maketitle

\begin{abstract}

  In this work, we extend the Constructive Cryptography (CC) framework introduced by Maurer in 2011 so as to handle interactive protocols.
  We design and construct a so-called {\em Interactive Server Memory Resource} (ISMR), that is an augmented version of the basic instantiation of a client-server protocol in CC, namely the Server Memory Resource.
%Our ISMR has enhanced funtionalities, such as...
We then apply our ISMR construction to two types of interactive cryptographic protocols for remote storage : Updatable Encryption (UE) and Private Information Retrieval (PIR).

Concerning UE, our results are a composable version of those protocols, clarifying the security guarantees achieved by {\em any} UE scheme. Namely, we give the relevant security notion to consider according to a given leakage context. Letting USMR denote our ISMR adapted to the UE application,
% , we prove that the security of the resource obtained by plugging (in a CC sense) any UE scheme to a USMR reduces to the security of the UE scheme with respect to :
%  \begin{itemize}
%  \item[] the IND-UE-CPA game, in the case when there is at least one leak per entry of the database per epoch;
%  \item[] the ENC-CPA+UPD-CPA games, in the any number of leaks case.
%    \end{itemize}
%    In other words,
we prove that $\mathsf{IND}\text{-}\mathsf{UE}\text{-}\mathsf{CPA}$ security is sufficient for a secure construction of a confidential USMR that hides the age of ciphertexts; and  $\mathsf{IND}\text{-}(\mathsf{ENC}+\mathsf{UPD})\text{-}\mathsf{CPA}$ security is sufficient for a secure construction of a confidential USMR in case of unrestricted leakage. As a consequence, contrary to what was claimed before, the $\mathsf{IND}\text{-}\mathsf{UE}$ security notion is not always stronger than the $\mathsf{IND}\text{-}(\mathsf{ENC+UPD})$ one.

  Another contribution to CC is a method of proof handling asymmetric challenges,
  namely game pairs of the form $(m,c)$, i.e. (plaintext, ciphertext), where the oracle either answers an encryption of $m$ or an update of $c$.
  Those occur for instance when proving $\mathsf{IND}\text{-}\mathsf{UE}\text{-}\mathsf{CPA}$ security.
  
  Concerning PIR, we also give a composable version of PIR protocols, yielding a unique model that unifies different notions of PIR : IT-PIR, C-PIR, one- or multi- server PIR. Using the flexibility of CC, we are also able to model PIR variants, such as SPIR.

\end{abstract}

\section{Introduction}
\label{sec:org301d508}

In this work, we are concerned with ensuring privacy in remote storage contexts. To do so, we start from the modeling of the client-server setting for non-interactive protocols in Constructive Cryptography (CC) by Badertscher and Maurer \cite{badertscher18_compos_robus_outsour_storag}.

CC is a composable framework that was introduced by Maurer \cite{maurer12_const_crypt_new_parad_secur_defin_proof}. Its \textit{top-down} approach makes it very intuitive and nice to work with. Badertscher and Maurer \cite{badertscher18_compos_robus_outsour_storag} first used it in the outsourced storage setting to clarify some security models and build protocols in a modular fashion. Jost \etal \cite{jost19_unified_compos_take_ratch} introduced the notion of \textit{global event history} to the theory and give a first treatment of adaptive security. This work was then pursued \cite{jost20_overc_impos_resul_compos_secur} by introducing interval-wise security guarantees in order to find a way to overcome impossibility results such as the commitment problem.

We extend this model by introducing and modeling interactivity in CC. We make this possible by proposing new ways of using of so-called {\em converters} and {\em simulators}.
We illustrate the power and the flexibility of our model on two privacy preserving and interactive schemes : Updatable Encryption (UE) and Private Information Retrieval (PIR). We make their modeling possible by introducing a new proof technique. We think that our treatment of interactivity in CC will permit to more easily build and to better understand the security guarantees of interactive protocols.
%that improve the scope and flexibility of CC.
%\emph{Constructive Cryptography}. 
\smallskip

\noindent
\emph{Updatable Encryption}. UE allows a client, who outsourced his encrypted data, to make an untrusty server update it. The huge real-life applications of such a functionality explains the recent renew of interest on the subject \cite{everspaugh17_key_rotat_authen_encry,lehmann18_updat_encry_post_compr_secur,klooss19_r_cca_secur_updat_encry_integ_protec,jiang20_direc_updat_encry_does_not_matter_much,boyd20_fast_secur_updat_encry}.
The concept and definition of UE first appeared in a paper by Boneh \etal in 2013 \cite{boneh13_key_homom_prfs_their_applic}, as an application of key homomorphic PRFs to the context of rotating keys for encrypted data stored in the cloud.

We are here interested in the \textit{ciphertext-independent} variant of UE, \ie, the one in which a token dependent only on the old and new keys is used to update all ciphertexts. This variant minimizes the communication complexity needed when doing key-rotation in the outsourced storage setting. Lehmann and Tackmann \cite{lehmann18_updat_encry_post_compr_secur} gave the first rigorous security treatment of those schemes by introducing the $\mathsf{IND}\text{-}\ENCCPA$ and $\mathsf{IND}\text{-}\UPDCPA$ security games alongside an Elgamal based secure UE scheme called \textsf{RISE}. Kloo{\ss} \etal \cite{klooss19_r_cca_secur_updat_encry_integ_protec} strengthened these security notions in the context of chosen-ciphertext security by introducing $(\mathsf{R})\mathsf{CCA}$ games and generic constructions on how to achieve UE schemes meeting their new definitions. Boyd \etal \cite{boyd20_fast_secur_updat_encry} introduced $\UECPA$, a strictly stronger security notion, in that it claims to hide the age of the data, and showed relations between all existing UE security notions. They also propose a generic construction for $\CCA$ secure schemes alongside the \textsf{SHINE} family of fast UE schemes based on random looking permutations and exponentiation. Finally, Jiang \cite{jiang20_direc_updat_encry_does_not_matter_much} showed that whether the token allowed to only upgrade or also downgrade ciphertexts did not make a difference in terms of security. Moreover, they propose the first post-quantum, LWE based, UE scheme. 
\smallskip

\noindent
\emph{Private Information Retrieval}. A PIR scheme \cite{CGKS95,chor98_privat_infor_retriev} involves one or more servers, holding a database with entries $x_1, ..., x_n$ and a client, who is interested in retrieving the $i$-th entry of the database, $i\in \lbrace, \ldots, n\rbrace$.
 The scheme allows the client to retrieve $x_i$ while hiding the index $i$ from the servers and using communication which is sublinear in $n$. There is a huge literature surrounding PIR, starting from the seminal work of Chor \etal \cite{CGKS95}.
 PIR comes in different flavors, one- or multi- server, information-theoretically secure PIR (IT-PIR) \cite{A97,CG97,BI01,BIKR02,Y08,E12,augot14_storag_effic_robus_privat_infor,DG16} and computationally secure PIR (cPIR) \cite{KO97,CMS99,KO00,GR05,OS07}. Note that IT-PIR can only be achieved with multiple servers.
 There are many more variants of PIR such as symmetric PIR \cite{GIKM00} where the client cannot learn any information about the entries of the database other than $x_i$ or batch PIR \cite{IKOS04,H16} whose aim is to amortize the computational cost of the servers over a batch of queries made by the client. In PIR, one can allow coalitions of servers if their size is under a given threshold.
 Finally, PIR has a lot of applications in cryptographic protocols : in private messaging \cite{SCM05,AS16,ACLS18}, online anonymity \cite{MOTBG11,KLDF16}, targeted advertising \cite{J01} and many more.

\subsection{Contributions}
% To CC
We extend the Constructive Cryptography model \cite{maurer12_const_crypt_new_parad_secur_defin_proof} so as to handle interactive protocols.

We call {\bf Interactive Server Memory Resource ($\ISMR$)} the {\em resource} we design, that constitutes the core of interactivity.
We make interactivity possible in the basic client-server setting of [1] by modifying the functionalities of converters : in addition to their basic features, they can now be used to modify the number or type of parameters of a query. As well, such converters have the ability to transiently deactivate an interface. We also introduce and model what we call {\em semi-honest} (commonly named honest-but-curious) interfaces, in which simulators can be plugged. This permits us to precisely model and describe the behaviour of semi-honest adversaries, as well as byzantine ones.

Another contribution to CC is a new improvement of a method of proof by Coretti \etal \cite{coretti13_const_confid_chann_authen_chann}: to prove the security of the different constructions, we define a sequence of hybrid systems, the first being the real one equipped with the protocol, the last the ideal one it is supposed to achieve. Then, we use a second hybrid argument between each consecutive systems of the sequence. Our proof thus evolves from a conjunction of small security-game proofs. Moreover, this allows us to handle an asymmetry : indeed, to prove $\mathsf{IND}\text{-}\mathsf{UE}\text{-}\mathsf{CPA}$ security, we have to handle game pairs of the form $(m,c)$, \ie (plaintext, ciphertext), where the oracle either answers an encryption of $m$ or an update of $c$. Those pairs cannot be handled by the classical game proofs proposed in \cite{coretti13_const_confid_chann_authen_chann}, which only treat requests of the form $(m,m')$ (mimicing their proofs in our case, the distinguishing advantage would be 1).

We then apply our interactivity in CC - in essence, our $\ISMR$ construction - to two types of cryptographic protocols for remote storage : UE and PIR.

%ToUE
Concerning UE, our results are a composable version of UE, clarifying the security guarantees achieved : Indeed, our modelization of UE in CC permits to give the security achieved by {\em any} UE scheme. Note that, in the only previous work we know of, that considers UE in the CC context \cite{FMM21}, the construction proposed was only valid for particular instantiations of UE schemes.
 More precisely, letting $\USMR$, for Updatable Server Memory Resource, denote our $\ISMR$ adapted to the UE application,
 our result is the following (theorem 1, rephrased): starting from a $\USMR$ on which we plug a converter modeling any UE scheme, we prove that the security of the new resource obtained reduces to the security of the UE protocol with respect to :
 \begin{itemize}
 \item[] the $\UECPA$ game, in the case when there is at least one leak per entry of the database per epoch;
 \item[] the $\mathsf{IND}\text{-}(\mathsf{ENC}+\mathsf{UPD})\text{-}\mathsf{CPA}$ games, in the any number of leaks case.
   \end{itemize}
   Said differently, we show that $\UECPA$ security is sufficient for a secure construction of a confidential $\USMR$ that hides the age of ciphertexts; and  $\mathsf{IND}\text{-}(\mathsf{ENC}+\mathsf{UPD})\text{-}\mathsf{CPA}$ security is sufficient for a secure construction of a confidential $\USMR$ in case of unrestricted leakage.
 
    Giving the relevant security notion to consider according to a given leakage context, as done above, we show that the $\mathsf{IND}\text{-}\mathsf{UE}$ security notion of Boyd \etal \cite{boyd20_fast_secur_updat_encry} is not always stronger than the one of \cite{lehmann18_updat_encry_post_compr_secur}, namely $\mathsf{IND}\text{-}(\mathsf{ENC}+\mathsf{UPD})$. 

 It is worth mentioning that our work on UE rules out the tedious commitment problem, that usually occurs when dealing with key exposures in composable frameworks.

 Concerning PIR, we :
 \begin{itemize}
 \item[-] give a composable version of PIR protocols;
  \item[-] propose a unique model, that unifies different notions of PIR : IT-PIR, C-PIR, one- or multi- server PIR. Indeed, the CC framework allows to handle IT-security contexts and computational ones at the same time. That is one notable difference with the UC framework;
  \item[-] are able to model PIR variants, thanks to the modularity of CC. For instance for the SPIR variant, using specification intersections, we can tune adversarial guarantees.
  \end{itemize}

\subsection{Organization of the paper}
\label{sec:org124d6ba}

Section \ref{sec:orge287bdc} presents the background needed to understand our work.
In section \ref{sec:ismr} we present our treatment of interactivity in the CC model.
In particular, we describe our Interactive Server-Memory Resource and our novel uses for converters and simulators.
We instantiate it to the UE setting, and that permits us to give a composable treatment of UE schemes, and to analyze their security properties w.r.t. different leakage contexts in section \ref{sec:org577f7f9}.
In section \ref{sec:pir}, we instantiate it to PIR protocols and give a composable and unified treatment of information-theoretic, computational and one or multi-server PIR.
In the multiple server case, we also model byzantine coalitions. Section \ref{sec:ccl} concludes the paper.

\section{Background}
\label{sec:orge287bdc}
\subsection{The Constructive Cryptography model}
\label{sec:orge0e968d}

% We give the required material to understand how we use the Constructive Cryptography framework by Maurer \etal \cite{maurer11_abstr_crypto,maurer12_const_crypt_new_parad_secur_defin_proof,maurer16_from_indis_to_const_crypt_and_back,jost19_unified_compos_take_ratch,jost20_overc_impos_resul_compos_secur}. We follow the presentation made by Jost and Maurer in \cite{jost20_overc_impos_resul_compos_secur}.

The CC model, introduced by Maurer \cite{maurer11_abstr_crypto} in 2011, and augmented since then \cite{maurer12_const_crypt_new_parad_secur_defin_proof,maurer16_from_indis_to_const_crypt_and_back,jost19_unified_compos_take_ratch,jost20_overc_impos_resul_compos_secur} aims at asserting the real security of cryptographic primitives. To do so, it redefines them in terms of so-called {\em resources} and {\em converters}.

In this model, starting from a basic resource (e.g. communication channel, shared key, memory server...), a converter (a cryptographic protocol) aims at constructing an enhanced resource, \ie, one with better security guarantees. The starting resource, lacking the desired security guarantees, is often called the \textit{real} resource and the obtained one is often called the \textit{ideal} resource, since it does not exist as is in the real world. An example of such an ideal resource is a confidential server, where the data stored by a client is readable by this client only. The only information that leaks to other parties is its length. This resource does not exist, but it can be emulated by an insecure server on which the client uses an encryption protocol where the encryption scheme is $\mathsf{IND}-\mathsf{CPA}$ secure. We say that this \textit{construction} of the confidential server is secure if the real world - namely, the insecure server together with the protocol - is \textit{just as good} as the ideal world - namely, the confidential server.
This means that, whatever the adversary can do in the real world, it could as well do in the ideal world. We use the fact that the ideal world is by definition secure and contraposition to conclude.

The CC model follows a top-down approach, allowing to get rid of useless hypotheses made in other models. A particularity of this model is its composability, in the sense that a protocol obtained by composition of a number of secure constructions is itself secure. In the following, we follow the presentation of \cite{jost20_overc_impos_resul_compos_secur}.

\subsubsection{Global Event History}
\label{sec:org91eefa4}

This work uses the globally observable events introduced in \cite{jost19_unified_compos_take_ratch}. Formally, we consider a \emph{global event history} \(\EE\) which is a list of event without duplicates. An event is defined by a name \(n\), and triggering the event \(n\) corresponds to the action of appending \(n\) to \(\EE\), denoted by \(\EE \xleftarrow{+} \EE_n\). For short, we use the notation \(\EE_n\) to say that event \(n\) happened. Finally, $\EE_{n} \prec \EE_{n'}$ means that the event $n$ precedes $n'$ in the event history.

\subsubsection{Resources, Converters and Distinguishers}

A \emph{resource} \(\R\) is a system that interacts, in a black-box manner, at one or more of its \emph{interfaces}, by receiving an input at a given interface and subsequently sending an output at the same interface. Do note that a resource only defines the observable behavior of a system and not how it is defined internally. The behavior of the resource depends on the global event history \(\EE\) and it can append events to it. We use the notation \([\R_1, \ldots, \R_k]\) to denote the parallel composition of resources. It corresponds to a new resource and, if \(\R_1, \ldots, \R_k\) have disjoint interfaces sets, the interface set of the composed resource is the union of those.
\smallskip

In CC, \emph{converters} are used to link resources and reprogram interfaces, thus expressing the local computations of the parties involved. A converter is plugged on a set of interfaces at the inside and provides a set of interfaces at the outside. When it receives an input at its outside interface, the converter uses a bounded number of queries to the inside interface before computing a value and outputting it at its outside interface.

A converter \(\pi\) connected to the interface set \(\mathcal{I}\) of a resource \(\R\) yields a new resource \(\R' := \pi^\mathcal{I}\R\). The interfaces of \(\R'\) inside the set \(\mathcal{I}\) are the interfaces emulated by \(\pi\). A protocol can be modelled as a tuple of converters with pairwise disjoint interface sets.
\smallskip

A \emph{distinguisher} \(\D\) is an environment that connects to all interfaces of a resource \(\R\) and sends queries to them. \(\D\) has access to the global event history and can append events that cannot be added by \(\R\). At any point, the distinguisher can end its interaction by outputting a bit. The advantage of a distinguisher is defined as

\[ \Delta^\D(\R, \SS) := \vert \Pr[\D^\EE(\R) = 1] - \Pr[\D^\EE(\SS) = 1] \vert, \]
\(\D^\EE\) meaning that the distinguisher has oracle access to the global event history \(\EE\).

\subsubsection{Specifications}
\label{sec:org5986b8e}

An important concept of CC is the one of \emph{specifications}. Systems are grouped according to desired or assumed properties that are relevant to the user, while other properties are ignored on purpose. A specification \(\mathcal{S}\) is a set of resources that have the same interface set and share some properties, for example confidentiality. In order to construct this set of confidential resources, one can use a specification of assumed resources \(\mathcal{R}\) and a protocol \(\pi\), and show that the specification \(\pi\mathcal{R}\) satisfies confidentiality. Proving security is thus proving that \(\pi\mathcal{R} \subseteq \mathcal{S}\), sometimes written as \(\mathcal{R} \xrightarrow{\pi} \mathcal{S}\), and we say that the protocol \(\pi\) constructs the specification \(\mathcal{S}\) from the specification \(\mathcal{R}\). The composition property of the framework comes from the transitivity of inclusion. Formally, for specifications \(\mathcal{R}, \mathcal{S}\) and \(\mathcal{T}\) and protocols \(\pi\) for \(\mathcal{R}\) and \(\pi'\) for \(\mathcal{S}\), we have $\mathcal{R} \xrightarrow{\pi} \mathcal{S} \wedge \mathcal{S}\xrightarrow{\pi'}\mathcal{T}\Rightarrow \mathcal{R}\xrightarrow{\pi'\circ\pi}\mathcal{T}$.

We use the real-world/ideal-world paradigm, and often refer to \(\pi\mathcal{R}\) and \(\mathcal{S}\) as the real and ideal-world specifications respectively, to understand security statements. Those statements say that the real-world is "just as good" as the ideal one, meaning that it does not matter whether parties interact with an arbitrary element of \(\pi\mathcal{R}\) or one of \(\mathcal{S}\). This means that the guarantees of the ideal specification \(\mathcal{S}\) also apply in the real world where an assumed resource is used together with the protocol.

Since specifications are set of resources, we can consider the intersection $\mathcal{S}\cap\mathcal{T}$ of two specifications $\mathcal{S}$ and $\mathcal{T}$. The resulting specification possesses the guarantees of both $\mathcal{S}$ and $\mathcal{T}$.

In this work, we use \emph{simulators}, \ie, converters that translate behaviors of the real world to the ideal world, to make the achieved security guarantees obvious. For example, one can model confidential servers as a specification \(\mathcal{S}\) that only leaks the data length, combined with an arbitrary simulator \(\sigma\), and show that \(\pi\mathcal{R} \subseteq \sigma\mathcal{S}\). It is then clear that the adversary cannot learn anything more that the data length.

\subsubsection{Relaxations}
\label{sec:org4da397b}

In order to talk about computational assumptions, post-compromise security or other security notions, the CC framework relies on \emph{relaxations} which are mappings from specifications to larger, and thus weaker, \emph{relaxed specifications}.
The idea of relaxation is that, if we are happy with constructing specification \(\mathcal{S}\) in some context, then we are also happy with constructing its relaxed variant. One common example of this is computational security. Let \(\epsilon\) be a function that maps distinguishers \(\D\) to the winning probability, in \([0,1]\), of a modified distinguisher \(\D'\) (the reduction) on the underlying computational problem. Formally,

   \begin{definition}
     Let $\epsilon$ be a function that maps distinguishers to a value in $[0,1]$. Then, for a resource $\R$, the reduction relaxation $\R^\epsilon$ is defined as
     %\vspace{-0.1cm}
    \[ \R^\epsilon := \lbrace \SS \mid \forall \D, \Delta^\D(\R, \SS) \leq \epsilon(\D)\rbrace  \]
This (in fact any) relaxation can be extended to a specification $\mathcal{R}$ by defining $\mathcal{R}^\epsilon := \cup_{\R\in \mathcal{R}}{\R^\epsilon}$.
   \end{definition}

   The other relaxation that we will use is the \textit{interval-wise relaxation} introduced in \cite{jost20_overc_impos_resul_compos_secur}. Given two predicates $P_1(\EE)$ and $P_2(\EE)$ on the global event history, the interval-wise relaxation $\R^{[P_1, P_2]}$ is the set of all resources that must behave like $\R$ in the time interval starting when $P_1(\EE)$ becomes true and ending when $P_2(\EE)$ becomes true. Outside this interval, we have no guarantees on how the resources behave.

   The two relaxations we use have nice composition properties, mainly they are compatible together and with parallel and sequential protocol applications, as shown in \cite{jost20_overc_impos_resul_compos_secur}. This means that all the constructions presented in this work can be used in a modular fashion inside bigger constructions, without needing to write a new security proof.

   \subsection{Updatable Encryption}
   \label{sec:orgda1dca6}

   \subsubsection{Definitions}

We follow the syntax given in \cite{boyd20_fast_secur_updat_encry}. A UE scheme is given by a tuple of algorithms $(\KG$, $\TG$, $\Enc$, $\Dec$, $\Upd)$ operating in epochs over message space \(\MS\), ciphertext space \(\CS\), key space \(\KS\) and token space \(\TS\). For a \emph{correct} UE scheme, these algorithms work as follows :

\begin{itemize}
\item On input \(1^\lambda\) where \(\lambda\) is a security parameter, \(\KG(1^\lambda)\in \KS\) returns a key.
\item On input \((k_e, k_{e+1})\in \KS^2\), \(\TG(k_e, k_{e+1})\in \TS\) returns a token \(\Delta_{e+1}\).
\item On input a plaintext \(m\in \MS\) and \(k_e \in \KS\), \(\Enc_{k_e}(m)\in \CS\) returns an encryption of \(m\) under the key \(k_e\) of epoch \(e\).
\item On input a ciphertext \(c\in \CS\) encrypting a message \(m\in \MS\) under a key \(k_e\in \KS\), \(\Dec_{k_e}(c)\in \MS\) returns the message \(m\). If \(c\) is an invalid ciphertext, a decryption error \(\diamond\) is returned instead.
\item On input a ciphertext \(c_e\in \CS\) encrypting a message \(m\in \MS\) under a key \(k_e\in \KS\) and a token \(\Delta_{e+1}\in \TS\) computed with the pair \((k_e, k_{e+1})\in \KS^2\), \(\Upd_{\Delta_{e+1}}(c)\in \CS\) returns ciphertext $c_{e+1}$. Correctness of the UE scheme requires that $c_{e+1}$ is an encryption of \(m\) under the key \(k_{e+1}\).
  %\vspace{-0.2cm}
\end{itemize}
For a precise definition of correctness, see \cite{lehmann18_updat_encry_post_compr_secur}.

\subsubsection{Security notions}
\label{sec:org2439924}

In this work, we study what confidentiality guarantees are exactly brought by the use of UE schemes. In previous works, the security of UE is described using security games. We will analyze the differences between the \(\mathsf{IND}\text{-}\mathsf{ENC} + \mathsf{IND}\text{-}\mathsf{UPD}\) security notion of \cite{lehmann18_updat_encry_post_compr_secur} and the \(\mathsf{IND}\text{-}\mathsf{UE}\) one of \cite{boyd20_fast_secur_updat_encry}. 
\smallskip

In \(\mathsf{IND}\text{-}\mathsf{UE}\) security, when given a plaintext \(m\) and a ciphertext \(c\) from a previous epoch encrypting a message of length $\vert m\vert$, the game challenges the adversary with either an encryption of $m$ or an update of $c$.

In \(\mathsf{IND}\text{-}\mathsf{ENC}\) security,  when given two plaintexts \(m_0\) and \(m_1\) of same length, the game challenges the adversary with an encryption of one of them.

In \(\mathsf{IND}\text{-}\mathsf{UPD}\) security, when given two old ciphertexts \(c_0\) and \(c_1\) encrypting two messages of the same length, the game challenges the adversary with an update of one of them.
\smallskip

We will study the confidentiality of UE schemes in the $\CPA$ setting of \cite{boyd20_fast_secur_updat_encry} where the adversary has access to both an encryption and an update oracle. We recall the different game oracles used in the $\CPA$ setting of \cite{boyd20_fast_secur_updat_encry} in fig.\ref{oracles} of the appendix.

\subsection{Private Information Retrieval}
\label{sec:pir_bg}

\subsubsection{Notation and definitions}

In this work, we model the content of a database as an array $\M$ of size $n$ over some alphabet $\Sigma$. A $k$-server PIR protocol involves $k$ servers $\S_1, \ldots, \S_k$ each holding a database $\M_j (1\leq j\leq k)$ that can be either $\M$ duplicated or a modified (\eg encrypted or encoded) version of $\M$. The user wants to retrieve some record $\M[i]$ for $1\leq i\leq n$, without revealing any information about $i$ to the servers. Formally, we have the following definition \cite{chor98_privat_infor_retriev}.

\begin{definition}[Private Information Retrieval]
A $k$-server PIR protocol is a triple $(\Q, \AA, \RR)$ of algorithms such that :
\begin{itemize}
        \item The user samples a random string $s$ with distribution $\mathcal{S}$, we denote this operation by $s \twoheadleftarrow \mathcal{S}$. The user can then use the algorithm $\Q$ to generate a $k$-tuple of queries $\Q(i, s) := (q_1, \ldots, q_k)$ and send each one to the corresponding server.
        \item Each server $\S_j$ computes an answer $\AA(j, \M_j, q_j) := a_j$ and sends it back to the user.
        \item The user can recover $\M[i]$ by invoking the reconstruction algorithm $\RR(a_1$, $\ldots$, $a_k$, $i$, $s)$
        \end{itemize}
        \end{definition}

Furthermore, the protocol is at least required to have the following {\em correctness} and {\em $t$-privacy} properties, where :

\begin{itemize}
        \item (Correctness) For all databases $\M\in \Sigma^n$, for every index $1\leq i\leq n$, the user recovers $\M[i]$, \ie, $\RR(\AA(1, \M_1, q_1), \ldots, \AA(k, \M_k, q_k), i, s) = \M[i]$ where $(q_1, \ldots, q_k) := \Q(i, s)$ and $s \twoheadleftarrow \mathcal{S}$.
        \item ($t$-Privacy) No coalition of at most $t$ servers can obtain any information about $i$.
\end{itemize}

\section{The Interactive Server Memory Resource}
\label{sec:ismr}

In \cite{badertscher18_compos_robus_outsour_storag}, Maurer and Badertscher propose an instantiation in CC of the client-server setting for non-interactive protocols. They notably introduce the basic Server Memory Resource (\(\SMR\) for short), where an honest client can read and write data on an outsourced memory owned by a potentially dishonest server. We extend their work to interactive protocols in the client-server setting by introducing a different type of \(\SMR\), namely the Interactive Server-Memory Resource (\(\ISMR\) for short). By interactive, we mean any protocol where the client sends a query and expects that the server will perform a given computation upon reception. For example, in computational PIR, the server is expected to answer with the result of a computation involving the whole database and the client's query.

One of the main differences with $\SMR$ is that, in \(\ISMR\), the server is now considered as semi-honest (through its interface \(\S\)), since it has to participate in the interactive protocol when receiving orders from the client. In \cite{badertscher18_compos_robus_outsour_storag}, the client was the only party to take part in the protocol, the role of the server being limited to storing data. In our case, the server agrees to apply a converter which carries out the computations requested by the client. We gather these capabilities in a sub-interface of \(\S\) denoted by \(\S.1\). Since the server is only semi-honest, we have no guarantees about how it will use its remaining capabilities. We gather those under the sub-interface \(\S.2\).

In the \(\ISMR\), we model interactions by letting the client interact with the server, through its interface \(\C\), according to a boolean value $\textsc{NeedInteraction}$. The server, through its sub-interface $\S.1$, is given the capability of checking this boolean to see if it needs to take actions, act if necessary and then switch the boolean value back. While the \(\ISMR\) is waiting for an interaction to end, \ie as long as $\textsc{NeedInteraction}$ is set to \(\texttt{true}\), the capabilities of the client at its interface are disabled. Our approach is general enough for the \(\ISMR\) to be used when modeling all kinds of interactive protocols in the outsourced storage setting. We can model interactive protocols that are information-theoretically, statistically or computationally secure, multi-client and multi-server protocols as well as passive or active adversaries. We showcase this by modeling UE schemes in sec.~\ref{sec:org577f7f9} and PIR schemes in sec.~\ref{sec:pir}. A formal description of the \(\ISMR\) is given in fig.~\ref{ISMR}.

\begin{figure}[htbp]
  \begin{mybox}{Resource $\ISMR_{\Sigma, n}$}
        \vspace{-0.4cm}
\begin{algorithm}[H]
\renewcommand{\thealgorithm}{}%no numbering
\makeatletter
\renewcommand*{\ALG@name}{Initialization}
\makeatother
\caption{Initialization}
\begin{algorithmic}%[1] for line numbers
	\State {$\textsc{NeedInteraction} \leftarrow \texttt{false}, \mathbb{M} \leftarrow [~]$}
\end{algorithmic}
\end{algorithm}

\vspace{-0.8cm}

\begin{algorithm}[H]
\renewcommand{\thealgorithm}{}%no numbering
\makeatletter
\renewcommand*{\ALG@name}{Interface}
\makeatother
\caption{Interface $\mathsf{C}$}
\begin{algorithmic}%[1] for line numbers
\Require {$(\texttt{read}, i)\in [n]$}
  \If {\textbf{not }\textsc{NeedInteraction}}
    \State {\Return $\M[i]$}
  \EndIf
\Require {$(\texttt{write}, i, x)\in [n] \times \Sigma$}
  \If {$\textbf{not }\textsc{NeedInteraction}$}
    \State {$\mathbb{M}[i] \leftarrow x$}
  \EndIf
\Require {$\texttt{askInteraction}$}
  \If {$\textbf{not }\textsc{NeedInteraction}$}
    \State {$\textsc{NeedInteraction} \leftarrow \texttt{true}$}
  \EndIf
\Require {$\texttt{getStatus}$}
  \State {\Return $\textsc{needInteraction}$}
\end{algorithmic}
\end{algorithm}

\vspace{-0.8cm}

\begin{algorithm}[H]
\renewcommand{\thealgorithm}{}%no numbering
\makeatletter
\renewcommand*{\ALG@name}{Sub-Interface}
\makeatother
\caption{Sub-Interface $\mathsf{S}.1$ of \textbf{Interface } $\mathsf{S}$}
\begin{algorithmic}%[1] for line numbers
\Require {$(\texttt{read}, i) \in [n]$}
  \State {\Return $\mathbb{M}[i]$}
\Require {$(\texttt{write}, i, x) \in [n]\times \Sigma$}
  \State {$\mathbb{M}[i] \leftarrow x$}
\Require {$\texttt{interact}$}
  \If {$\textsc{needInteraction}$}
    \State {$\textsc{needInteraction} \leftarrow \texttt{false}$}
  \EndIf
\end{algorithmic}
\end{algorithm}

\vspace{-0.8cm}

\begin{algorithm}[H]
\renewcommand{\thealgorithm}{}%no numbering
\makeatletter
\renewcommand*{\ALG@name}{Sub-Interface}
\makeatother
\caption{Sub-Interface $\mathsf{S}.2$ of \textbf{Interface } $\mathsf{S}$}
\begin{algorithmic}%[1] for line numbers
\Require {$(\texttt{leak}, i)\in [n]$ \quad // other adversarial capabilities can be added}
  \State {\Return $\mathbb{M}[i]$}
\Require {$\texttt{getStatus}$}
  \State {\Return $\textsc{needInteraction}$}
\end{algorithmic}
\end{algorithm}
\end{mybox}
\caption{The interactive Server-Memory Resource with finite alphabet $\Sigma$ and size $n$. The sub-interface $\S.1$ guarantees that the server follows a protocol through the application of a converter. On the other hand, no guarantees are given at the sub-interface $\S.2$. \label{ISMR}}
\end{figure}

%%% Supprimer cette partie en gardant les premières phrases des deux paragraphes et en développant les idées puis en mettant see PIR sec, see UE sec
We are able to model all these different protocols by taking advantage of the power and the flexibility of the CC framework together with our new $\ISMR$. A key point of our work is to propose new ways to use converters and simulators in CC.

First, we show how one can use a converter to modify the number and/or the types of arguments required when sending a query. For example, in PIR, a client wants to retrieve the $i$-th entry of a database without revealing $i$ to the server. Thus, the client will not send $i$ directly to the server. Instead, the client will perform a computation involving $i$ and some secret and send the result $q$ to the server where $q$ belongs to some set $S$. The real resource must account for this by implementing a query $(\texttt{query}, q)\in S$ at the client's interface. The issue is that the final goal of the protocol, modeled in the ideal resource, is to retrieve the $i$-th entry of a database of size $n$, which needs to be modeled as a query $(\texttt{query}, i)\in [1, n]$. We bridge the gap between the real and ideal resources by allowing the client's converter (used in the real-world) to reprogram the $(\texttt{query}, q)\in S$ query into a $(\texttt{query}, i)\in [1, n]$ one. Moreover, we also allow the converter to increase or decrease the number of arguments of a query. See sec.~\ref{sec:org577f7f9} for more details.

Secondly, we show how one can use a converter to disable a capability at an interface. This is particularly useful when dealing with semi-honest adversaries. Indeed, such an adversary participates in the protocol by carrying computations for the client. This is modeled by the application of a converter. For example, in UE schemes, the server is expected to update ciphertexts for the client. To carry out its computation, the server must be able to retrieve the update token sent by the client. This is modeled by a $(\texttt{fetchToken}, e)$ query. In UE schemes, this operation is not considered to be malicious. Thus, we add a $(\texttt{leakToken}, e)$ query to model a malicious access to the update token (with a different behavior than the $(\texttt{fetchToken}, e)$ query). Then, we need to disable the $(\texttt{fetchToken}, e)$ query to prevent an adversary from accessing the update token without querying $(\texttt{leakToken}, e)$. Since our modelization of interactivity allows us to use converters on semi-honest interfaces, we can use a converter to disable the $(\texttt{fetchToken}, e)$ query at the outside interface of the server. See sec.~\ref{sec:pir} for more details.
%%% FIN de la section à virer

To build an \(\ISMR\) with stronger security guarantees, we will use the construction notion of the CC framework, as used in \cite{badertscher18_compos_robus_outsour_storag}. This notion is illustrated in fig.~\ref{usmr_constr}. One difference with the work of \cite{badertscher18_compos_robus_outsour_storag} is that, although there is a protocol plugged in interface \(\S\), this interface is only semi-honest and doesn't belong to the so-called "honest parties". We thus need to plug a simulator at this interface in the ideal world if we hope to achieve any meaningful construction. One way to see this is to consider the \(\texttt{read}\) capabilities at interfaces \(\C\) and \(\S\) when an encryption scheme, such as UE, is used. Since the client holds the decryption key, the interface \(\C\) is able to retrieve the plaintexts corresponding to the ciphertexts stored in memory. This is not true for the interface \(\S\) since it does not have access to the key under normal circumstances. Since there is no encryption in the ideal world, we need a simulator to simulate ciphertexts when \(\S\) sends \(\texttt{read}\) requests at its interface, as otherwise distinguishing between the two worlds would be trivial.

\begin{figure}[htbp]
  %\vspace{0.6cm}
\tikzset{every picture/.style={line width=0.75pt}} %set default line width to 0.75pt        
\begin{tikzpicture}[x=0.75pt,y=0.75pt,yscale=-0.9,xscale=0.9]
\draw    (158.67,50.33) -- (158.5,35) ;
\draw   (56.5,85.5) -- (86.33,85.5) -- (86.33,115.33) -- (56.5,115.33) -- cycle ;
\draw    (56.83,99) -- (44.83,99) ;
\draw    (86.33,99) -- (109,98.8) ;
\draw   (108.67,50.33) -- (208.33,50.33) -- (208.33,150) -- (108.67,150) -- cycle ;
\draw   (253,72.5) -- (222.17,72.5) -- (222.17,103.33) -- (253,103.33) -- cycle ;
\draw    (252.67,87) -- (261.67,87) ;
\draw    (222.17,87) -- (208.5,86.8) ;
\draw    (409.67,50.33) -- (409.5,35) ;
\draw    (337.33,101) -- (360,100.8) ;
\draw   (359.67,50.33) -- (459.33,50.33) -- (459.33,150) -- (359.67,150) -- cycle ;
\draw    (503.67,80) -- (518.67,80) ;
\draw    (473.17,81) -- (459.5,80.8) ;
\draw   (503.17,71.17) -- (503.06,131.16) -- (473.62,131.1) -- (473.74,71.11) -- cycle ;
\draw    (502.67,121.13) -- (518.67,121.13) ;
\draw    (459.67,121.13) -- (473.67,121.13) ;
\draw    (208.02,117.13) -- (261.02,117.13) ;
\draw    (261.02,117.13) -- (261.02,87.13) ;
\draw    (261.02,100.13) -- (270.02,100.13) ;
\draw    (518.67,80) -- (518.67,121.13) ;
\draw    (518.67,100.57) -- (528.65,100.43) ;
\draw (158.5,100.17) node   [align=left] {$\ISMR$};
\draw (159,27) node   [align=left] {$\W$};
\draw (71,100) node   [align=left] {$\mathsf{prot}$};
\draw (35.33,100) node   [align=left] {$\C$};
\draw (303,100) node   [align=left] {$\approx $};
\draw (237.5,86) node   [align=left] {$\mathsf{prot}$};
\draw (271.02,92.13) node [anchor=north west][inner sep=0.75pt]   [align=left] {$\S$};
\draw (409.5,100.17) node   [align=left] {Stronger\\ \:$\ISMR$};
\draw (410,27) node   [align=left] {$\W$};
\draw (327.33,102) node   [align=left] {$\C$};
\draw (518,66) node [anchor=north west][inner sep=0.75pt]   [align=left] {$\S.1$};
\draw (518.67,119.13) node [anchor=north west][inner sep=0.75pt]   [align=left] {$\S.2$};
\draw (475,92) node [anchor=north west][inner sep=0.75pt]   [align=left] {$\mathsf{sim}$};
\draw (260,72) node [anchor=north west][inner sep=0.75pt]   [align=left] {$\S.1$};
\draw (261,115) node [anchor=north west][inner sep=0.75pt]   [align=left] {$\S.2$};
\draw (530.02,92.13) node [anchor=north west][inner sep=0.75pt]   [align=left] {$\S$};
\end{tikzpicture}
\vspace{-0.3cm}
\caption{The construction notion for $\ISMR$. On the left, the plain $\ISMR$ equipped with a protocol. On the right, the stronger $\ISMR$ equipped with a simulator. The construction is deemed secure if there exists a simulator for which the two systems are indistinguishable.\label{usmr_constr}}
\vspace{-0.4cm}
\end{figure}
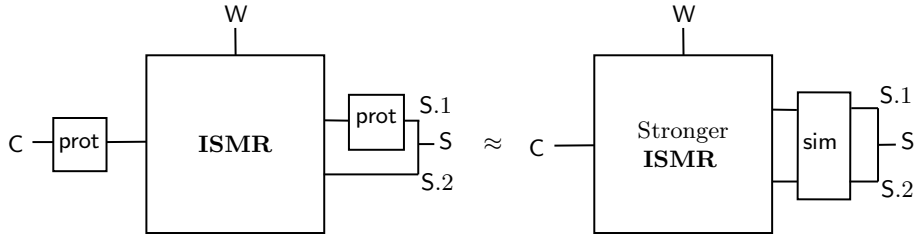

For the sake of clarity, the $\ISMR$ will be renamed $\USMR$ (for Updatable Server Memory-Resource) when modeling UE schemes, see sec.~\ref{sec:org577f7f9}. When modeling PIR schemes, the $\ISMR$ will be renamed $\DB$ (for Database), see sec.~\ref{sec:pir}. We will also rename $\textsc{NeedInteraction}$, $\texttt{askInteraction}$ and $\texttt{interact}$. 

\section{A composable treatment of Updatable Encryption}
\label{sec:org577f7f9}

\subsection{Instantiation of the $\ISMR$ to Updatable Encryption}
\label{sec:org216d202}

We recall that the $\ISMR$ is renamed $\USMR$ in this section. 
Moreover, at interface \(\S\), we also distinguish between honestly reading the memory  - through \(\texttt{read}\) requests at sub-interface \(\S.1\) - to update a ciphertext without trying to use this information against the client; and maliciously reading (we prefer to say \textit{leaking}) the memory - through \(\texttt{leak}\) requests at sub-interface \(\S.2\) - to gain information and try to break the confidentiality guarantees of the client. Sending a $(\texttt{leak}, i)$ request triggers the event $\EE^{\leaked}_{\Data, i}$.
A full description of the $\USMR$ is given in fig.~\ref{USMR}.

\begin{figure}[htbp]
  \begin{mybox}{Resource $\USMR_{\Sigma, n}$}
    \vspace{-0.4cm}
\begin{algorithm}[H]
\renewcommand{\thealgorithm}{}%no numbering
\makeatletter
\renewcommand*{\ALG@name}{Initialization}
\makeatother
\caption{Initialization}
\begin{algorithmic}%[1] for line numbers
	\State {$\textsc{NeedUpdate} \leftarrow \texttt{false}$}
    \State {$\M \leftarrow [~]$}
\end{algorithmic}
\end{algorithm}

\vspace{-0.8cm}

\begin{algorithm}[H]
\renewcommand{\thealgorithm}{}%no numbering
\makeatletter
\renewcommand*{\ALG@name}{Interface}
\makeatother
\caption{Interface $\C$}
\begin{algorithmic}%[1] for line numbers
\Require {$(\texttt{read}, i)\in [n]$}
  \If {\textbf{not }\textsc{NeedUpdate}}
    \State {\Return $\M[i]$}
  \EndIf
\Require {$(\texttt{write}, i, x)\in [n] \times \Sigma$}
  \If {$\textbf{not }\textsc{NeedUpdate}$}
    \State {$\M[i] \leftarrow x$}
  \EndIf
\Require {$\texttt{askUpdate}$}
  \If {$\textbf{not }\textsc{NeedUpdate}$}
    \State {$\textsc{NeedUpdate} \leftarrow \texttt{true}$}
  \EndIf
\Require {$\texttt{getStatus}$}
  \State {\Return $\textsc{needUpdate}$}
\end{algorithmic}
\end{algorithm}

\vspace{-0.8cm}

\begin{algorithm}[H]
\renewcommand{\thealgorithm}{}%no numbering
\makeatletter
\renewcommand*{\ALG@name}{Sub-Interface}
\makeatother
\caption{Sub-Interface $\S.1$ of \textbf{Interface } $\S$}
\begin{algorithmic}%[1] for line numbers
\Require {$(\texttt{read}, i) \in [n]$}
  \State {\Return $\M[i]$}
\Require {$(\texttt{write}, i, x) \in [n]\times \Sigma$}
  \State {$\M[i] \leftarrow x$}
\Require {$\texttt{update}$}
  \If {$\textsc{needUpdate}$}
    \State {$\textsc{needUpdate} \leftarrow \texttt{false}$}
  \EndIf
\end{algorithmic}
\end{algorithm}

\vspace{-0.8cm}

\begin{algorithm}[H]
\renewcommand{\thealgorithm}{}%no numbering
\makeatletter
\renewcommand*{\ALG@name}{Sub-Interface}
\makeatother
\caption{Sub-Interface $\S.2$ of \textbf{Interface } $\S$}
\begin{algorithmic}%[1] for line numbers
\Require {$(\texttt{leak}, i)\in [n]$}
  \State {$\EE \xleftarrow{+} \EE^{\leaked}_{\Data, i}$}
  \State {\Return $\M[i]$}
\Require {$\texttt{getStatus}$}
  \State {\Return $\textsc{needUpdate}$}
\end{algorithmic}
\end{algorithm}
\end{mybox}
\caption{The $\ISMR$ viewed as an updatable server-memory resource $\USMR$ with finite alphabet $\Sigma$ and memory size $n$. Interface $\S$ guarantees that it will endorse an honest behavior, through the application of a converter, at its sub-interface $\S.1$. However, no such guarantees are offered at its sub-interface $\S.2$.\label{USMR}}
\end{figure}

%\vspace{-0.2cm}
\subsection{The Updatable Key Resource \(\UpdKey\)}
\label{sec:orgee2477c}

We have so far described the \(\USMR\), a Server-Memory Resource to which we want to apply a UE scheme to strengthened the security guarantees of the client. To do so, we need to model the use of cryptographic keys and update tokens needed by UE schemes. This is why we introduce an Updatable Key Resource, named \(\UpdKey\), whose role is to model the existence, the operations and the availability of keys as well as update tokens. In the following, let \(\K\) be the key space of UE schemes. Given \(k\) and \(k'\) two keys in \(\K\), the notation \(\Delta \leftarrow \T(k, k')\) denotes the assignation, to the variable \(\Delta\), of the token that updates ciphertexts encrypted under the key \(k\) to ones encrypted under \(k'\).

\begin{figure}[htbp]
  \begin{mybox}{Resource $\UpdKey$}
        \vspace{-0.4cm}
\begin{algorithm}[H]
\renewcommand{\thealgorithm}{}%no numbering
\makeatletter
\renewcommand*{\ALG@name}{Initialization}
\makeatother
\caption{Initialization}
\begin{algorithmic}%[1] for line numbers
	\State {$e \leftarrow 1, \KK \leftarrow [~], \TT \leftarrow [\perp]$}
    \State {$k \twoheadleftarrow \K, \KK \leftarrow \KK \mid\mid k$}
    \State {$\EE \xleftarrow{+} \EE^{\epoch}_e$}
\end{algorithmic}
\end{algorithm}

\vspace{-0.8cm}

\begin{algorithm}[H]
\renewcommand{\thealgorithm}{}%no numbering
\makeatletter
\renewcommand*{\ALG@name}{Interface}
\makeatother
\caption{Interface $\C$}
\begin{algorithmic}%[1] for line numbers
\Require {$\texttt{fetchKey}$}
    \State {\Return $\KK[e]$}
\Require{$\texttt{nextEpoch}$}
    \State {$k_{e+1} \twoheadleftarrow \K$}
    \State {$\KK \leftarrow \KK \mid\mid k_{e+1}$}
    \State {$\Delta_{e+1} \leftarrow \T(k_e, k_{e+1})$}
    \State {$\TT \leftarrow \TT \mid\mid \Delta_{e+1}$}
    \State {$e \leftarrow e + 1$}
    \State {$\EE \xleftarrow{+} \EE^{\epoch}_e$}
    \State {\Return $k_{e+1}$}
\end{algorithmic}
\end{algorithm}

\vspace{-0.8cm}

\begin{algorithm}[H]
\renewcommand{\thealgorithm}{}%no numbering
\makeatletter
\renewcommand*{\ALG@name}{Sub-Interface}
\makeatother
\caption{Sub-Interface $\S.1$ of \textbf{Interface } $\S$}
\begin{algorithmic}%[1] for line numbers
\Require {$(\texttt{fetchToken}, i)$}
    \If {$2\leq i\leq e$}
        \State {\Return $\TT[i]$}
    \Else
        \State {\Return $\perp$}
    \EndIf
\end{algorithmic}
\end{algorithm}

\vspace{-0.8cm}

\begin{algorithm}[H]
\renewcommand{\thealgorithm}{}%no numbering
\makeatletter
\renewcommand*{\ALG@name}{Sub-Interface}
\makeatother
\caption{Sub-Interface $\S.2$ of \textbf{Interface } $\S$}
\begin{algorithmic}%[1] for line numbers
\Require {$(\texttt{leakKey}, i)$}
    \If {$i\leq e$ \textbf{ and } $\EE^\leaked_{\Key, i}$}
        \State {\Return $\KK[i]$}
    \Else
        \State {\Return $\perp$}
    \EndIf
\Require {$(\texttt{leakToken}, i)$}
    \If {$2\leq i\leq e$ \textbf{ and } $\EE^\leaked_{\Token, i}$}
        \State {\Return $\TT[i]$}
    \Else
        \State {\Return $\perp$}
    \EndIf
\end{algorithmic}
\end{algorithm}
\end{mybox}
\caption{The updatable key (with its associated token) resource $\UpdKey$. For interface $\S$, we use the same distinction between its sub-interfaces $\S.1$ and $\S.2$ as in the $\USMR$ description.\label{updkey}}
\end{figure}

% In our model, \(\S\) subdivides its interface into \(\S.1\), where it applies a converter (\ie a protocol) that guarantees a honest behavior at this interface, and sub-interface \(\S.2\), where the server gives absolutely no guarantees and can freely endorse a malicious behavior.
In $\UpdKey$, the \(\texttt{fetchToken}\) request is always accessible at sub-interface \(\S.1\), since the protocol used at this interface will prevent the information it provides to be maliciously used, whereas the request \(\texttt{leakToken}\) at interface \(\S.2\) requires that a special event \(\EE^\leaked_{\Token,i}\) has been triggered before returning the update token to epoch \(i\), which can be used for malicious purposes.
These events are triggered by the environment which, in CC, can be given an interface that is usually denoted by \(\W\) for world interface.

This separation between \texttt{read}/\texttt{leak} in $\USMR$ and \texttt{fetchToken}/\texttt{leakToken} requests in $\UpdKey$ is important because it allows us to describe the security guarantees of the system more precisely. Indeed, since the server needs to retrieve the token to update the ciphertexts, if we consider that this token "leaked", it becomes impossible to express the post-compromise security guarantees brought by UE schemes. This is because if all tokens leak, a single key exposure compromises the confidentiality of ciphertexts for all subsequent epochs.

\subsection{An Updatable Encryption protocol for the \(\USMR\)}
\label{sec:org25759eb}

The main focus of this section being the study of the security guarantees brought by the use of UE schemes, we have to define an encryption protocol based on UE for the \(\USMR\) (and $\UpdKey$) and study its effects. Since we work in the CC framework, we describe our protocol as a pair of converters \((\uec, \ues)\) where \(\uec\) will be plugged in interface \(\C\) and \(\ues\) in sub-interface \(\S.1\) of the \(\USMR\) and $\UpdKey$. A formal description of \(\uec\) (resp. \(\ues\)) can be found in fig.~\ref{uec_converter} (resp. fig.~\ref{ues_converter}). 

\begin{figure}[htbp]
  \vspace{-0.4cm}
  \begin{mybox}{Converter $\uec$}
        \vspace{-0.4cm}
\begin{algorithm}[H]
\renewcommand{\thealgorithm}{}%no numbering
\makeatletter
\renewcommand*{\ALG@name}{Initialization}
\makeatother
\caption{Initialization}
\begin{algorithmic}%[1] for line numbers
	\State {$k \leftarrow \texttt{fetchKey}$ at interface $\C$ of $\UpdKey$}
\end{algorithmic}
\end{algorithm}

\vspace{-0.8cm}

\begin{algorithm}[H]
\renewcommand{\thealgorithm}{}%no numbering
\makeatletter
\renewcommand*{\ALG@name}{Interface}
\makeatother
\caption{Interface $\texttt{out}$}
\begin{algorithmic}%[1] for line numbers
\Require {$(\texttt{read}, i)\in [n]$}
  \State {$c \leftarrow (\texttt{read}, i)$ at interface $\C$ of $\USMR$}
  \If {$c\neq \perp$}
    \State {$m \leftarrow \Dec_k(c)$}
    \If {$m \neq \diamond$}
      \State {\Return $m$}
    \EndIf
  \EndIf
\Require {$(\texttt{write}, i, x)\in [n] \times \Sigma$}
  \State {$c \leftarrow \Enc_k(x)$}
  \State {Send $(\texttt{write}, i, c)$ at interface $\C$ of $\USMR$}
\Require {$\texttt{askUpdate}$}
  \State {$u \leftarrow \texttt{getStatus}$ at interface $\C$ of $\USMR$}
  \If {\textbf{not} $u$}
    \State {$k \leftarrow \texttt{nextEpoch}$ at interface $\C$ of $\UpdKey$}
    \State {Send $\texttt{askUpdate}$ at interface $\C$ of $\USMR$}
  \EndIf
  \State
\end{algorithmic}
\end{algorithm}
\end{mybox}
\caption{The client's converter $\uec$ for UE scheme $(\KG, \TG, \Enc, \Dec, \Upd)$ with decryption error symbol $\diamond$.\label{uec_converter}}
\vspace{-0.2cm}
\end{figure}

\begin{figure}[htbp]
  \begin{mybox}{Converter $\ues$}
        \vspace{-0.4cm}
\begin{algorithm}[H]
\renewcommand{\thealgorithm}{}%no numbering
\makeatletter
\renewcommand*{\ALG@name}{Initialization}
\makeatother
\caption{Initialization}
\begin{algorithmic}%[1] for line numbers
	\State {$\Delta \leftarrow \perp, \: e \leftarrow 1$}
\end{algorithmic}
\end{algorithm}

\vspace{-0.8cm}

\begin{algorithm}[H]
\renewcommand{\thealgorithm}{}%no numbering
\makeatletter
\renewcommand*{\ALG@name}{Interface}
\makeatother
\caption{Interface $\texttt{out}$}
\begin{algorithmic}%[1] for line numbers
\Require {$\texttt{update}$}
  \State {$u \leftarrow \texttt{getStatus}$ at interface $\S$ of $\USMR$}
  \If {$u$}
    \State {$e \leftarrow e + 1$}
    \State {$\Delta \leftarrow (\texttt{fetchToken}, e)$ at interface $\S$ of $\UpdKey$}
    \For {$i = 1\ldots n$}
      \State {$c \leftarrow (\texttt{read}, i)$ at interface $\S$ of $\USMR$}
      \State {Send $(\texttt{write}, i, \Upd_\Delta(c))$ at interface $\S$ of $\USMR$}
    \EndFor
  \State {Send $\texttt{update}$ at interface $\S$ of $\USMR$}
  \EndIf
\Require {$(\texttt{read}, i)$, $(\texttt{write}, i, x)$, or $(\texttt{fetchToken}, e)$}
  \State {\Return $\perp$}
\end{algorithmic}
\end{algorithm}
\end{mybox}
\caption{\label{ues_converter}The server's converter $\ues$ for UE scheme $(\KG, \TG, \Enc, \Dec, \Upd)$. Since the server is semi-honest, the converter monitors the behavior of the requests at sub-interface $\S.1$. The server guarantees that updates are done correctly (using the UE scheme) through the $\texttt{update}$ request, and that the requests $\texttt{read}, \texttt{write}$ and $\texttt{fetchToken}$ are only used to update the ciphertexts and not to gain information to break the confidentiality of the data. These requests are thus disabled at its interface $\texttt{out}$ but they can still be used internally by the converter.}
\vspace{-0.6cm}
\end{figure}
\vspace{-0.2cm}
\subsection{The confidential \(\USMR\)}
\label{sec:org4d2836d}

The security guarantees of the \(\USMR\) can be improved by requiring confidentiality for the client's data. The resulting resource is called confidential \(\USMR\) and we will refer to it as \(\cUSMR\). In practice, this means that, on a \((\texttt{leak}, i)\) request at interface \(\S\), only the length of $\M[i]$ is returned to the adversary and not the \(i\)-th entry itself. The \texttt{read} and \texttt{write} capabilities of sub-interface \(\S.1\) are removed. The resource \(\cUSMR\) is described in fig.~\ref{cUSMR}.

\begin{figure}[htbp]
 \vspace{-0.5cm}
  \begin{mybox}{Resource $\cUSMR$}
        \vspace{-0.4cm}
\begin{algorithm}[H]
\renewcommand{\thealgorithm}{}%no numbering
\makeatletter
\renewcommand*{\ALG@name}{Sub-Interface}
\makeatother
\caption{Sub-Interface $\S.1$ of $\textbf{Interface } \S$}
\begin{algorithmic}%[1] for line numbers
\Require {$\texttt{update}$}
  \If {\textsc{NeedUpdate}}
    \State {$\textsc{NeedUpdate} \rightarrow \texttt{false}$}
  \EndIf
\end{algorithmic}
\end{algorithm}

\vspace{-0.8cm}

\begin{algorithm}[H]
\renewcommand{\thealgorithm}{}%no numbering
\makeatletter
\renewcommand*{\ALG@name}{Sub-Interface}
\makeatother
\caption{Sub-Interface $\S.2$ of $\textbf{Interface } \S$}
\begin{algorithmic}%[1] for line numbers
\Require {$(\texttt{leak}, i)\in [n]$}
  \State {\Return $\vert \M[i]\vert$}
\Require {$\texttt{getStatus}$}
  \State {\Return \textsc{NeedUpdate}}
\end{algorithmic}
\end{algorithm}
\end{mybox}
\vspace{-0.2cm}
\caption{The confidential and  updatable server-memory resource $\cUSMR$. Only differences with $\USMR$ are shown.\label{cUSMR}}
\vspace{-0.4cm}
\end{figure}

At this point, we rectify the claim made by Boyd \etal \cite{boyd20_fast_secur_updat_encry}, namely that their $\mathsf{IND}\text{-}\mathsf{UE}$ security notion is better than the $\mathsf{IND}\text{-}\mathsf{ENC}\text{+}\mathsf{UPD}$ notions, in that it hides the age of ciphertexts. By age, we mean the last epoch in which the ciphertext was freshly written to the database. We show that this is only true when ciphertexts can leak at most one time per epoch. Indeed, if a ciphertext can leak at least two times per epoch, the adversary can use its first (resp. second) leak at the start (resp. end) of each epoch. If the ciphertext has changed between the two leaks, it must have been rewritten during this epoch and its age is now $0$. If it has not changed, then its age is incremented. We see that in this setting, the age of ciphertexts cannot be protected.

In the rest of this work, we will distinguish between resources that only allow one leak per ciphertext per epoch, denoted by a $1$ in an exponent (\eg $\USMR^1$), and resources that allow two or more leaks per ciphertext per epoch, denoted by a $+$ (\eg $\USMR^+$). If the number of leaks doesn't matter we will omit the exponent.

It is important to stress that the security of UE schemes is thought of in the context of adaptive adversaries, where the use of a UE scheme should bring post-compromise confidentiality guarantees to the client. In this work, we thus consider adaptive adversaries and use the extension of CC developed by Jost \etal in \cite{jost20_overc_impos_resul_compos_secur}.

\subsection{Handling post-compromise security}
\label{sec:org2fe4f1d}

The goal of this section is to give an exact description of the post-compromise security guarantees given by UE schemes. Said differently, we want to explain how the security guarantees evolve after a key exposure. When dealing with situations such as key exposures, composable frameworks usually stumble on an impossibility result called the \emph{commitment problem}. This problem is the following : given a message \(m\), how can an online simulator explain a simulated ciphertext \(c\), generated without knowledge of \(m\), with a key \(k\) such that \(c\) decrypts to \(m\) under this key. Thanks to a recent work of Jost \etal \cite{jost20_overc_impos_resul_compos_secur}, the CC framework is well equipped to deal with this impossibility result. This is done through the use of interval-wise security guarantees. In CC, the interval-wise relaxation describes security guarantees within an interval delimited by predicates on the global event history. For example, we can describe security guarantees before and after the key leaks.
\smallskip

A UE scheme is said to have \emph{uni-directional} updates if an update token can only move a ciphertext from the old key to the new key. A UE scheme supports \emph{bi-directional} updates if the update token can additionally downgrade ciphertexts from the new key to the old key. Jiang \cite{jiang20_direc_updat_encry_does_not_matter_much} recently proved that schemes supporting uni-directional updates do not bring more security compared to those with bi-directional updates, in the sense that the security notions for uni and bi-directional updates are proved to be equivalent. Thus, in what follows, we only focus on schemes with bi-directional updates.
\smallskip

In UE schemes, it is clear that the confidentiality of the user data is lost when an epoch key leaks. This security breach remains in subsequent epochs if the keys continue to leak or if successive update tokens leak. However, as soon as we encounter an epoch where neither the key nor the update token leaks, the confidentiality is restored. This remains true until a future epoch, where either a key leaks or consecutive update tokens leak until a key is finally exposed. This is due to the fact that ciphertexts can be upgraded and downgraded with update tokens to an epoch where a key is exposed.

In the epoch timeline, the areas where confidentiality is preserved are called \emph{insulated regions}. They have been studied and used in previous works \cite{lehmann18_updat_encry_post_compr_secur,klooss19_r_cca_secur_updat_encry_integ_protec,boyd20_fast_secur_updat_encry,jiang20_direc_updat_encry_does_not_matter_much}. We describe those regions with their extreme left and right epochs. These pairs of epochs are called \emph{firewalls}. We recall the definition used in \cite{boyd20_fast_secur_updat_encry}.

   \begin{definition}
An insulated region with firewalls $\fwl$ and $\fwr$ is a consecutive sequence of epochs $(\fwl, \ldots, \fwr)$ for which :
   \begin{enumerate}
       \item No key in the sequence of epochs $(\fwl, \ldots, \fwr)$ is corrupted.
       \item The tokens $\Delta_\fwl$ and $\Delta_{\fwr + 1}$ are not corrupted, if they exist.
       \item All tokens $(\Delta_{\fwl + 1}, \ldots, \Delta_\fwr)$ are corrupted.
   \end{enumerate}
The set of all firewall pairs is denoted by $\FW$. The set of all insulated regions is denoted by $\IR := \cup_{(\fwl, \fwr)\in \FW}{\lbrace \fwl, \ldots, \fwr\rbrace }$.
   \end{definition}

The epochs where the confidentiality guarantees do not hold are the ones not found in \(\IR\). The set of firewalls \(\FW\) can easily be described using predicates on the global event history. This is done in the following manner.
\begin{align*}
 \FW := \lbrace (\fwl, \fwr) \mid & \fwl \leq \fwr, \\
                                 & \forall e\in \lbrace \fwl, \ldots, \fwr\rbrace, \neg\EE^\leaked_{\Key, e}, \\
                                 & \neg\EE^\leaked_{\Token, \fwl} \text{ and } \neg\EE^\leaked_{\Token, \fwr + 1}, \\
                                 & \forall e\in \lbrace \fwl + 1, \ldots, \fwr\rbrace, \EE^\leaked_{\Token, e}\rbrace
\end{align*}

We say that the $i$-th entry of a database $\M$ is compromised if an adversary gets an encryption of $\M[i]$ in an epoch $e$ such that $e$ does not belong to an insulated region. Formally, we introduce the predicate
\begin{align*}
  P_{\texttt{compromised}, i}(\EE) :=\: & \exists e : \EE^\epoch_e \prec \EE^\leaked_{\Data,i} \prec \EE^\epoch_{e+1} \quad \wedge \\
  & \forall (\fwl, \fwr)\in \FW, \neg(\EE^\epoch_\fwl \prec \EE^\epoch_e \prec \EE^\epoch_{\fwr+1})
\end{align*}
  
The right side of the conjunction means `the epoch $e$ does not belong to an insulated region`. Recall that the event $\EE^\epoch_e$ is triggered by $\UpdKey$ on a $\texttt{nextEpoch}$ request and the event $\EE^\leaked_{\Data, i}$ is triggered by the $\USMR$ on a $(\texttt{leak}, i)$ request.

Then, we introduce the event $\EE^\insec_j$ which indicates that the $j$-th entry of the database is not confidential. This event can only be triggered by the environment. Now, we can modify our definition of the $\cUSMR$ in the following way : on a $(\texttt{leak}, j)$ request, this resource now returns $\M[j]$ if $\EE^\insec_j$ has been triggered and $\vert \M[j]\vert$ otherwise. For $1\leq i\leq n$, we introduce the predicate
\[ P_{\mathtt{only}, i}(\EE) := \bigwedge_{j\in \lbrace 1,\ldots,n\rbrace \setminus \lbrace i\rbrace}{\EE^\insec_j} \]
It formalizes that we do not consider the confidentiality of plaintexts other than the $i$-th one.

Following the notation of \cite{jost20_overc_impos_resul_compos_secur}, we introduce our main theorem. Let $n$ be the number of entries stored in the server. Our construction is an intersection of $n$ specifications. For $i\in \lbrace 1, \ldots, n\rbrace$, we assume that an adversary knows every entry of the database except the $i$-th one. Then, the $i$-th specification guarantees the confidentiality of the $i$-th entry until it trivially leaks because an adversary gained access to an encryption of this plaintext under an exposed epoch key.

\begin{theorem}
  \label{thUE}
     Let $\Sigma$ be a finite alphabet and $n\in \N$. There exists a sequence of simulators $(\sigma_i)_{1\leq i\leq n}$ such that the protocol $\pi_{\UE} := (\uec, \ues)$, described in fig.~\ref{uec_converter} and \ref{ues_converter}, based on an $\UECPA$ secure UE scheme constructs the $\cUSMR_{\Sigma, n}$ from the $\USMR_{\Sigma, n}$ and $\UpdKey$
 %    \vspace{-0.2cm}
      \[ [\USMR_n^1, \UpdKey] \xrightarrow{\pi_{\UE}} \bigcap_{1\leq i\leq n}{(\sigma_i \cUSMR_n^1)^{[P_{\mathtt{only}, i}(\EE),P_{\mathtt{compromised}, i}(\EE)]:\epsilon_{\CPA}}}  \]
      where $\epsilon_\CPA$ denotes our reduction, given in theorem \ref{thCPA}, from distinguishing between the $\cUSMR^1$ with our simulator and the $\USMR^1$ with our protocol, to winning the $\UECPA$ game.
   \end{theorem}
   
In the above theorem, we can replace $(\USMR^1$, $\cUSMR^1$, $\epsilon_\CPA$, $\UECPA)$ with $(\USMR^+$, $\cUSMR^+$, $\epsilon_\CPAp$, $\ENCCPA + \UPDCPA)$ when we deal with unrestricted leakage to include all our results.
   
\begin{proof}
  Since we consider an intersection of $n$ specification, we need to prove $n$ constructions. For $i\in \lbrace 1, \ldots, n\rbrace$, we need to prove that their exists a simulator $\sigma_i$ such that, in the interval $[P_{\mathtt{only}, i}(\EE),P_{\mathtt{compromised}, i}(\EE)]$, the protocol $\pi_{\UE}$ constructs the $\cUSMR^1_{\Sigma, n}$ from the $\USMR^1_{\Sigma, n}$ and $\UpdKey$ with respect to $\sigma_i$. This construction is formalized and proven in theorem \ref{thCPA} of sec.~\ref{sec:org63a24c5} where we give a detailed reduction from breaking our construction to winning the  $\UECPA$ game.
%  The simulator initially generates an epoch key. For all the messages except the $i$-th one, the ideal-world resource (the $\cUSMR^1$) leaks the actual message. Hence, the simulator can encrypt it himself and a future leak of a possibly updated version of this ciphertext together with the corresponding epoch key would not create a commitment problem.
%  For the $i$-th message, the simulator only knows its length and it can only encrypt a random plaintext of the same length. We give a description of the simulator along with a detailed reduction from distinguishing the real-world and the ideal-world to winning the $\mathsf{IND}\text{-}\mathsf{UE}\text{-}\mathsf{CPA}$ game in sec.~\ref{sec:org63a24c5}.

  When replacing $(\USMR^1$, $\cUSMR^1$, $\epsilon_\CPA$, $\UECPA)$ with $(\USMR^+$, $\cUSMR^+$, $\epsilon_\CPAp$, $\ENCCPA + \UPDCPA)$ in the above theorem, we need two prove $n$ more constructions. We formalize and prove these construction in theorem \ref{thENCUPD} of sec.~\ref{sec:orgbbe97e7} where we give a detailed reduction from breaking our construction to winning the  $\ENCCPA + \UPDCPA$ games.
\end{proof}

\subsection{Exact security of UE schemes}
\label{sec:orge661c8e}

\subsubsection{At most one leak per entry per epoch : the \(\CPA\) case}
\label{sec:org63a24c5}

In this section, we work with a $\USMR^1$ of size $n$ where the attacker can leak entries of the database at its interface \(\S\). This capability allows the distinguisher (which is connected to every interface of the system) to build an encryption oracle. Indeed, the distinguisher can use the client interface \(\C\) to write messages of its choice into the database, and then leak the ciphertexts associated to these messages by sending a \texttt{leak} request at interface \(\S\). This fact  motivates the use of a \(\CPA\) security notion since it's tailored to bring security in the presence of such an encryption oracle.

\paragraph{The simulator $\sigma_{k, \CPA}$ \label{simCPA}}
Since we are trying to prove our theorem \ref{thUE}, we place ourselves in the context of this theorem : we take $k\in \lbrace 1, \ldots, n\rbrace$ and we place ourselves in the interval $[P_{\mathtt{only}, k(\EE)}, P_{\mathtt{compromised}, k(\EE)}]$ where we do not consider the confidentiality of plaintexts other than the $k$-th one. In this case, the simulator \(\sigma_{k, \CPA}\) works as follows. It simulates the epoch keys and tokens and, on a \((\texttt{leakKey}, i)\) or \((\texttt{leakToken}, i)\) request, it checks if the event \(\EE^\leaked_{\Key, i}\) (respectively \(\EE^\leaked_{\Token, i}\)) exists in the Global Event History, and leaks the corresponding epoch key (respectively token) to the adversary if it is the case, and \(\perp\) otherwise. On a \((\texttt{leak}, k)\) request, the simulator forwards it to the ideal resource to get a length $\ell$ and returns a fresh encryption of a random plaintext of length $\ell$ under the current epoch key. Finally, on a \((\texttt{leak}, i)\) request, $i\neq k$, the simulator forwards it to the ideal resource to get a plaintext $x$ and returns a fresh encryption of $x$ under the current epoch key.

\subsubsection{\(\UECPA\) security is sufficient for a secure construction of \(\cUSMR\) that hides the age}
\label{sec:org1d4dd77}

The fact that $\UECPA$ is sufficient to construct $\cUSMR^1$ from $\USMR^1$ and $\UpdKey$ is expressed in th.~\ref{thUE} through an intersection of specifications. The following theorem shows how we construct each of those specifications.

\begin{theorem}
  \label{thCPA}
  Let $\Sigma$ be a finite alphabet, $n\in \N$ and $k\in \lbrace 1, \ldots, n\rbrace$. The protocol $\mathsf{ue} := (\uec, \ues)$ described in figures \ref{uec_converter} and \ref{ues_converter} based on a UE scheme constructs the $\cUSMR^1_{\Sigma, n}$ from the basic $\USMR^1_{\Sigma, n}$ and $\UpdKey$ inside the interval $[P_{\mathtt{only}, k}(\EE), P_{\mathtt{compromised}, k}(\EE)]$, with respect to the simulator $\sigma_{k, \CPA}$ described in \ref{simCPA} and the dummy converter $\honSrv$ (that disables any adversarial behavior). More specifically, we construct reductions such that, for all distinguishers $\D$ in a set of distinguishers $\DD$,
  \[ \Delta^\D(\honSrv^\S\uec^\C \ues^\S [\USMR^1_{\Sigma, n}, \UpdKey], \honSrv^\S\cUSMR^1_{\Sigma, n}) = 0  \]
\begin{align*}
  \Delta^\D(\uec^\C \ues^\S  [\USMR^1_{\Sigma, n}, & \UpdKey], \sigma_{k, \CPA}^{\S}\cUSMR^1_{\Sigma, n}) \leq  \\
  & (2q+r)\cdot\sup_{\D'\in \DD}\Delta^{\D'}(\G_\UPD^\UECPA, \G_\ENC^\UECPA)
  \end{align*}
where $q$ (resp. $r$) is an upper bound on the number of writes (resp. updates) made by the distinguisher to the memory location $k$.
\end{theorem}

The first condition, called \textit{availability}, checks if the two systems behave in the same way when no adversary is present. It rules out trivial protocols that would ensure confidentiality by not writing data in memory for example. In all this work, \textit{availability} follows from the correctness of the schemes used. For clarity and conciseness, we will omit it in the proofs.

\begin{proof}
  Let \(\R := \enc^\C \upd^\S [\USMR^1, \UpdKey]\) be the the real system and \(\I := \sigma_{k, \CPA}^{\S} \cUSMR^1\) be the ideal system. The two systems behave in the same way except when leaking the content of $\M[k]$ : $\R$ leaks an encryption of $\M[k]$ while $\I$ leaks an encryption of a random plaintext of length $\vert \M[k]\vert$. In order to determine the advantage of a distinguisher in distinguishing \(\R\) from \(\I\), denoted by \(\Delta^\D(\R, \I)\), we proceed with a sequence of systems. We introduce a hybrid system \(\SS\), then we determine the distinguishing advantages \(\Delta^\D(\R, \SS)\) and \(\Delta^\D(\SS, \I)\), the triangular inequality allowing us to bound \(\Delta^\D(\R, \I)\) by the sum of those two advantages.
\medskip

\noindent
\textbullet \(\:\)  Let \(\SS\) be a resource that behaves just like \(\R\) except on query \((\texttt{leak}, k)\) where it leaks an encryption of a random plaintext of length $\vert \M[k]\vert$ instead of an encryption of \(\M[k]\), if \(\M[k]\) contains a fresh encryption and not an updated one. This happens if a query \((\texttt{write}, k, x)\) has been issued by the client in the current epoch. In the case when \(\M[k]\) contains an updated version of a ciphertext, the two resources behave in the exact same way.

Let \(q\) be an upper bound on the number of $(\texttt{write}, k, .)$ queries issued to the systems. We define a hybrid resource \(\H_i\) that behaves just like \(\R\) on the first \(i\) ($\texttt{write}, k, .$) queries and like \(\SS\) afterwards. Then we define a reduction \(\CC_i\) that behaves like \(\H_i\) except it uses the game \(\G^{\ENCCPA}_b\) oracles instead of doing the UE operations by itself and on the \(i\)-th $(\texttt{write}, k, .)$ request (of the form \((\texttt{write}, k, x)\)) it challenges the game with input \((x, \bar{x})\), with $\bar{x}$ random of length $\vert x\vert$, to receive the ciphertext. We have
%\vspace{-0.1cm}
\[ \R \equiv \H_q \text{ and } \SS \equiv \H_0 \]
and
\[ \H_i \equiv \G_0^\ENCCPA \CC_i \equiv \G_1^\ENCCPA \CC_{i+1} \]

Indeed, this can be seen on the following timeline (\ref{tab:orgc966621}).

\begin{table}[htbp]
\centering
\begin{tabular}{|c|c|c|c|c|}
\hline
\(j\)-th $(\texttt{write}, k, .)$ query & \(j<i\) & \(j=i\) & \(j=i+1\) & \(j>i+1\)\\
\hline
\({\color{green}\G_0^\ENCCPA} \CC_i\) & \(\Enc(x)\) & \color{green}{$\Enc(x)$} & \(\Enc(\bar{x})\) & \(\Enc(\bar{x})\)\\
\hline
\({\color{green}\G_1^\ENCCPA} \CC_{i+1}\) & \(\Enc(x)\) & \(\Enc(x)\) & \color{green}{$\Enc(\bar{x})$} & \(\Enc(\bar{x})\)\\
\hline
\end{tabular}
\caption{\label{tab:orgc966621}Leakage behavior of both systems for each \((\texttt{write}, k, .)\) request.}
\end{table}

Let \(\CC_I\) be a reduction that samples \(i\in [1, q]\) at random and behaves like \(\CC_i\) and define \(\D' := \D\CC_I\). We have,
%\vspace{-0.2cm}
\[ \Pr[\D'(\G_0^\ENCCPA) = 1] = \frac{1}{q} \sum_{i = 1}^q{\Pr[\D(\CC_i \G_0^\ENCCPA) = 1]} \]
\vspace{-0.1cm}
and
\vspace{-0.1cm}
\begin{align*}
 \Pr[\D'(\G_1^\ENCCPA) = 1] & = \frac{1}{q} \sum_{i=1}^q{\Pr[\D(\CC_i \G_1^\ENCCPA) = 1]} \\
 & = \frac{1}{q} \sum_{i=0}^{q-1}{\Pr[\D(\CC_i \G_0^\ENCCPA) = 1]}
\end{align*}

Finally, the advantage of the distinguisher in distinguishing system \(\R\) from \(\SS\) is
\begin{align*}
  \Delta^\D (\R, \SS) & = \Delta^\D(\H_q, \H_0) \\
                      & = \Delta^\D(\CC_q \G_0^\ENCCPA, \CC_0 \G_0^\ENCCPA) \\
                      & = \vert \Pr[\D(\CC_q \G_0^\ENCCPA) = 1] - \Pr[\D(\CC_0 \G_0^\ENCCPA) = 1] \vert \\
                      & = \vert \sum_{i=1}^q{\Pr[\D(\CC_i \G_0^\ENCCPA) = 1]} - \sum_{i=0}^{q-1}{\Pr[\D(\CC_i \G_0^\ENCCPA) = 1]} \vert \\
                      & = q \cdot \vert \Pr[\D'(\G_0^\ENCCPA) = 1] - \Pr[\D'(\G_1^\ENCCPA) = 1] \vert \\
                      & = q\cdot \Delta^{\D'}(\G_0^\ENCCPA, \G_1^\ENCCPA)
\end{align*}

\smallskip

\noindent
\textbullet \(\:\) Let us consider the systems \(\SS\) and \(\I\). By definition, \(\SS\) behaves just like \(\I\) except on query \((\texttt{leak}, k)\) where it leaks an updated ciphertext of an encryption of \(\M[k]\) (instead of a fresh encryption of a random \(\bar{x}\) of length $\vert \M[k]\vert$ in the ideal system) if \(\M[k]\) contains an updated encryption and not a fresh one. In the case when \(\M[k]\) contains a fresh ciphertext, the two resources behave in the exact same way. Namely, they leak an encryption of a random \(\bar{x}\) of length $\vert \M[k]\vert$.

Let \(r\) be an upper bound on the number of update queries issued to the systems. We define a hybrid resource \(\H'_i\) that behaves just like \(\SS\) on the first \(i\) update queries and like \(\I\) afterwards. Then we define a reduction \(\CC'_i\) that behaves like \(\H'_i\) except it uses the game \(\G^{\UECPA}_{\mathsf{xxx}}\), where \(\mathsf{xxx} \in \lbrace\ENC, \UPD\rbrace\), oracles instead of doing the UE operations by itself and on the \(i\)-th update computation for the encryption $c$ of $\M[k]$, it challenges the game with input \((\bar{x}, c)\) to receive either a fresh encryption of the random plaintext \(\bar{x}\) (of length $\M[k]$) or the updated version of the ciphertext \(c\). We have
%\vspace{-0.1cm}
\[ \SS \equiv \H'_r \text{ and } \I \equiv \H'_0 \]

and
\[ \H'_i \equiv \G_\UPD^\UECPA \CC'_i \equiv \G_\ENC^\UECPA \CC'_{i+1} \]

Indeed, this can be seen on the following timeline (\ref{tab:orgf802954}) 

\begin{table}[htbp]
\centering
\begin{tabular}{|c|c|c|c|c|}
\hline
\(j\)-th update query for $\M[k]$ & \(j<i\) & \(j=i\) & \(j=i+1\) & \(j>i+1\)\\
\hline
\({\color{green}\G_\UPD^\UECPA} \CC'_i\) & \(\Upd(c)\) & \color{green}{$\Upd(c)$} & \(\Enc(\bar{x})\) & \(\Enc(\bar{x})\)\\
\hline
\({\color{green}\G_\ENC^\UECPA} \CC'_{i+1}\) & \(\Upd(c)\) & \(\Upd(c)\) & \color{green}{$\Enc(\bar{x})$} & \(\Enc(\bar{x})\)\\
\hline
\end{tabular}
\caption{\label{tab:orgf802954}Leakage behavior of both systems for each update request ($\bar{x}$ is always a random plaintext of length $\M[k]$).}
\end{table}

Let \(\CC'_I\) be a reduction that samples \(i\in [1, r]\) at random and behaves like \(\CC'_i\) and define \(\D'' := \D\CC'_I\). We have,
%\vspace{-0.1cm}
\[ \Pr[\D''(\G_\UPD^\UECPA) = 1] = \frac{1}{r} \sum_{i = 1}^r{\Pr[\D(\CC'_i \G_\UPD^\UECPA) = 1]} \]
%\vspace{-0.2cm}
and

\begin{align*}
\Pr[\D''(\G_\ENC^\UECPA) = 1] & = \frac{1}{r} \sum_{i=0}^{r-1}{\Pr[\D(\CC'_i \G_\UPD^\UECPA) = 1]}
\end{align*}

Finally, the advantage of the distinguisher in distinguishing system \(\SS\) from \(\I\) is

\begin{align*}
  \Delta^\D (\SS, \I) & = \Delta^\D(\H'_r, \H'_0) \\
                      & = \Delta^\D(\CC'_r \G_\UPD^\UECPA, \CC'_0 \G_\UPD^\UECPA) \\
                      & = \vert \Pr[\D(\CC'_r \G_\UPD^\UECPA) = 1] - \Pr[\D(\CC'_0 \G_\UPD^\UECPA) = 1] \vert \\
                      & = \vert \sum_{i=1}^r{\Pr[\D(\CC'_i \G_\UPD^\UECPA) = 1]} - \sum_{i=0}^{r-1}{\Pr[\D(\CC'_i \G_\UPD^\UECPA) = 1]} \vert \\
                      & = r \cdot \vert \Pr[\D''(\G_\UPD^\UECPA) = 1] - \Pr[\D''(\G_\ENC^\UECPA) = 1] \vert \\
                      & = r\cdot \Delta^{\D''}(\G_\UPD^\UECPA, \G_\ENC^\UECPA)
\end{align*}
\smallskip

\noindent
\textbullet \(\:\) We use the triangular inequality to conclude. Let \(q\) be our upper bound on the number of writes and \(r\) be our upper bound on the number of updates. The advantage of the distinguisher in distinguishing the real system \(\R\) from the ideal one \(\I\) is
\begin{align*}
\Delta^\D(\R, \I) & \leq \Delta^\D(\R, \SS) + \Delta^\D(\SS, \I) \\
                  & = q\cdot \Delta^{\D'}(\G_0^\ENCCPA, \G_1^\ENCCPA) +  r\cdot \Delta^{\D''}(\G_\UPD^\UECPA, \G_\ENC^\UECPA) \\
                  & = 2q\cdot \Delta^{\D'\CC''}(\G_\UPD^\UECPA, \G_\ENC^\UECPA) + r\cdot \Delta^{\D''}(\G_\UPD^\UECPA, \G_\ENC^\UECPA) \\
                  & \leq (2q + r) \cdot \Delta^{\mathcal{D}}(\G_\UPD^\UECPA, \G_\ENC^\UECPA)\footnotemark
\end{align*}

\footnotetext{Since we can also split the games the other way and consider $\UPDCPA$ security first, we can replace $2q+r$ with $\min\lbrace 2q+r, q+2r\rbrace$.}

Where the reduction \(\CC''\) is given by Boyd \etal in \cite{boyd20_fast_secur_updat_encry} to prove the following 

 \begin{proposition}
Let $\Pi$ be a UE scheme. For any $\ENCCPA$ adversary $\A$ against $\Pi$, there exists a reduction $\CC''$ such that
  \[ \Delta^{\A}(\G^\ENCCPA_0, \G^\ENCCPA_1) \leq 2\cdot\Delta^{\A\CC''}(\G^\UECPA_\ENC, \G^\UECPA_\UPD) \]
 \end{proposition}

 We also use the notation \(\Delta^{\mathcal{D}}(\mathbf{X}, \mathbf{Y}) = \sup_{\D\in \mathcal{D}}\Delta^\D(\mathbf{X}, \mathbf{Y})\). \(\square\)
 \end{proof}

 \begin{remark}
   We point out that our interval choice, our simulator and our reduction circumvent the commitment problem. Indeed, when a key exposure makes us leave an insulated region, the adversary can :
   \begin{enumerate}
   \item decrypt the content of ciphertexts stored at location $i\neq k$. These ciphertexts are perfectly simulated since we do not consider the confidentiality of their plaintexts in our interval. Thus, there is no commitment problem.
   \item decrypt the content of the $k$-th ciphertext before its content has been used to produce the $\mathsf{CPA}$ game challenge. This closes our interval since the $k$-th plaintext is no longer confidential.
   \item decrypt the content of the $k$-th ciphertext after its content has been used to produce the $\mathsf{CPA}$ game challenge. This triggers a trivial win condition in the $\mathsf{CPA}$ game, the adversary thus loses the game and our interval closes like above.
   \end{enumerate}
 \end{remark}

\subsubsection{Any number of leaks : the \(\CPA\) case}
\label{sec:org3bdbee5}

This time, we are proving our theorem \ref{thUE} in the context of unrestricted leakage. Just like before, let $n$ be the size of the $\USMR^+$ and take $k\in \lbrace 1, \ldots, n\rbrace$. We place ourselves in a time interval where we do not consider the confidentiality of messages other than the $k$-th one. In this case the simulator \(\sigma_{k, \CPAp}\) works as follows. \label{simCPA2} It simulates the epoch keys and tokens and, on a \((\texttt{leakKey}, i)\) or \((\texttt{leakToken}, i)\) request, it checks if the event \(\EE^\leaked_{\Key, i}\) (respectively \(\EE^\leaked_{\Token, i}\)) exists in the Global Event History and leaks the corresponding epoch key (respectively token) to the adversary if it is the case and \(\perp\) otherwise. The simulator uses the ideal history to know which entries of the database correspond to fresh encryptions or updated encryptions. Together with its simulated epoch keys and tokens, this allows the simulator to maintain a simulated memory (and a simulated history) where fresh encryptions of $\M[k]$ (in the real world) are replaced with fresh encryptions of random plaintexts of length $\vert \M[k]\vert$ and updated ciphertexts encrypting $\M[k]$ (in the real world) are replaced with updates of ciphertexts of random plaintexts of length $\vert \M[k]\vert$. When $i\neq k$, the simulated memory and history perfectly match their real-world counterparts. Finally, on a \((\texttt{leak}, i)\) request, the simulator returns the \(i\)-th entry of its simulated memory.

\subsubsection{\(\ENCCPA+\UPDCPA\) security are sufficient for a secure construction of \(\cUSMR^+\)}
\label{sec:orgbbe97e7}
The fact that $\ENCCPA + \UPDCPA$ is sufficient to construct $\cUSMR^+$ from $\USMR^+$ and $\UpdKey$ is expressed in th.~\ref{thUE} through an intersection of specifications. The following theorem shows how we construct each of those specifications.

\begin{theorem}
  \label{thENCUPD}
  Let $\Sigma$ be a finite alphabet, $n\in \N$ and $k\in \lbrace 1, \ldots, n\rbrace$. The protocol $\mathsf{ue} := (\uec, \ues)$ described in figures \ref{uec_converter} and \ref{ues_converter} based on a UE scheme constructs the $\cUSMR^+_{\Sigma, n}$ from the basic $\USMR^+_{\Sigma, n}$ and $\UpdKey$ inside the interval $[P_{\mathtt{only}, k}(\EE), P_{\mathtt{compromised}, k}(\EE)]$, with respect to the simulator $\sigma_{k, \CPAp}$ described in \ref{simCPA2}. More specifically, we construct reductions $\CC'$ and $\CC''$ such that, for all distinguishers $\D$,

  \begin{align*}
  \Delta^\D(\uec^\C \ues^\S & [\USMR^+_{\Sigma, n},\UpdKey], \sigma_{k, \CPAp}^{\S}\cUSMR^+_{\Sigma, n}) \leq  \\ & q\cdot\Delta^{\D\CC'}(\G_0^\ENCCPA, \G_1^\ENCCPA) + r\cdot \Delta^{\D\CC''}(\G_0^\UPDCPA, \G_1^\UPDCPA)
  \end{align*}

  where $q$ (resp. $r$) is an upper bound on the number of writes (resp. updates) made by the distinguisher to the memory location $k$.
\end{theorem}

The proof is very similar to the one detailed in the \(\UECPA\) case. It can be found in appendix \ref{sec:org41788ad}. 

In \cite{boyd20_fast_secur_updat_encry} the authors argued that \(\UECPA\) security is stronger than \(\ENCCPA + \UPDCPA\). Since we showed that \(\ENCCPA + \UPDCPA\) security is sufficient to securely construct the \(\cUSMR^+\) from the \(\USMR^+\) equipped with the \((\enc, \upd)\) converters, we conclude that \(\UECPA\) security cannot be necessary for this secure construction in the unrestricted leakage model. This notion is thus too strong in this setting.

\section{A composable and unified treatment of Private Information Retrieval}
\label{sec:pir}

\subsection{Our modelization of PIR in CC}

Recall that the $\ISMR$ is renamed $\DB$ to model PIR. The resource $\DB$ is described in fig.~\ref{DB}. It has the following interfaces :

\begin{itemize}
      \item Interface $\C_0$ : The user (or client) can use this interface to initialize the state of the database. When the user is done with this interface, he can send the $\texttt{initComplete}$ request which sets the $\textsc{Active}$ boolean to $\texttt{true}$. This turns off the interface $\C_0$ and unlocks all the others interfaces of the resource. This interface is useful to model PIR protocols working on encrypted or encoded data.
      \item Interface $\C$ : The user main interface. We rename $\texttt{askInteraction}$ into $\mathtt{query}(q)$ which sends a query $q$ to the resource. When the servers all respond, the user can also use $\mathtt{reconstruct}()$ to recover their answers.
      \item Interface $\S$ : The interface of the server, who will adopt a semi-honest (also called honest-but-curious) adversarial behavior. This means that the server accepts to plug a converter in its interface to handle $\texttt{answer}$ requests (we renamed $\texttt{interact}$ to $\texttt{answer}$) according to the PIR protocol. But the server will also try to use the information it has access to, mainly the requests of the clients (through $\texttt{getQuery}$ requests), the content of the database (through $\texttt{read}$ requests) and a log file of all honest actions (through $\texttt{getHist}$), in order to break the privacy guarantees of the clients.      
\end{itemize}

For simplicity we will present a database that can only process a single query and answer. To construct a database with multiples queries, one can compose multiple single-use databases in parallel or parameterize the database with the number of query treated. In order to give security guarantees for an unbounded-query database, one has to show that the security guarantees hold for any number of queries. We also limit the number of clients to one for simplicity but, as most PIR protocols support an arbitrary number of clients, it is also possible to increase the number of clients by duplicating the interface $\C$ or by letting each client connect to the same interface.

We will give a straight generalization of $\DB$ to multiple servers later on.

\begin{figure}[htbp]
\begin{mybox}{Resource $\DB_{\Sigma, n}$}
\vspace{-0.4cm}
  \begin{algorithm}[H]
\renewcommand{\thealgorithm}{}%no numbering
\makeatletter
\renewcommand*{\ALG@name}{Initialization}
\makeatother
\caption{Initialization}
\begin{algorithmic}%[1] for line numbers
	\State $\textsc{Init}, \textsc{Active} \leftarrow \texttt{false}$, $\M \leftarrow [~], \Hist \leftarrow [~]$, $q \leftarrow  \perp, a\leftarrow \perp$
\end{algorithmic}
\end{algorithm}

\vspace{-0.8cm}

\begin{multicols}{2}
\begin{algorithm}[H]
\renewcommand{\thealgorithm}{}%no numbering
\makeatletter
\renewcommand*{\ALG@name}{Interface}
\makeatother
\caption{Interface $\C_0$}
\begin{algorithmic}%[1] for line numbers
\Require {$(\texttt{init},  \M') \in \Sigma^n$}
         \If {\textbf{not} \Init}
             \State {$\M \leftarrow \M'$, $\Hist \leftarrow \Hist \mid\mid (0, \mathtt{init})$}
             \State {$\Hist \leftarrow \Hist \mid\mid (0, \mathtt{init})$}
             \State {$\Init \leftarrow \texttt{true}$}
         \EndIf

\Require {$(\texttt{read}, i) \in [1,n]$}
	\If {\Init \textbf{ and not} \Active}
	    \State {$\Hist \leftarrow \Hist \mid\mid (0, \texttt{R}, i)$}
	    \State {\Return $\M[i]$}
	\EndIf

\Require {$(\texttt{write}, i, x) \in [1,n]\times \Sigma$}
	\If {\Init \textbf{ and not} \Active}
            \State {$\Hist \leftarrow \Hist \mid\mid (0, \texttt{W}, i, x)$}
	    \State {$\M[i] \leftarrow x$}
	\EndIf

\Require {$\texttt{initComplete}$}
	\If {\Init \textbf{ and not} \Active}
            \State {$\Active \leftarrow \texttt{true}$}
	\EndIf	
\end{algorithmic}
\end{algorithm}

\columnbreak

\begin{algorithm}[H]
\renewcommand{\thealgorithm}{}%no numbering
\makeatletter
\renewcommand*{\ALG@name}{Interface}
\makeatother
\caption{Interface $\C$}
\begin{algorithmic}%[1] for line numbers
\Require {$(\texttt{query}, q')$}
	\If {\Active}
        \If {$q = \perp$}
            \State {$q \leftarrow q'$}
            \State {$\Hist \leftarrow \Hist \mid\mid (\texttt{query}, q)$}
        \EndIf
    \EndIf

\Require {$\texttt{reconstruct}$}
         \If {\Active}
             \If {$a \neq \perp$}
                 \State {\Return $a$}
             \EndIf
         \EndIf
\end{algorithmic}
\end{algorithm}

\vspace{-0.8cm}

\begin{algorithm}[H]
\renewcommand{\thealgorithm}{}%no numbering
\makeatletter
\renewcommand*{\ALG@name}{Interface}
\makeatother
\caption{Interface $\S$}
\begin{algorithmic}%[1] for line numbers
\Require {$(\texttt{answer}, a')$}
            \If {$q \neq \perp$ \textbf{ and } $a = \perp$}
	        \State {$a \leftarrow a'$}
                \State {$\Hist \leftarrow \Hist \mid\mid (\texttt{answer}, a)$}
            \EndIf
\Require {$(\texttt{read}, i) \in [1, n]$}
    \State {\Return $\M[i]$}
\Require {$\texttt{getQuery}$}
    \State {\Return $q$}
\Require {$\texttt{getHist}$}
    \State {\Return $\Hist$}
\end{algorithmic}
\end{algorithm}
\end{multicols}
\end{mybox}
\caption{\label{DB}Our $\ISMR$ viewed as a basic interactive database $\DB$ with size $n$ and alphabet $\Sigma$.}
\end{figure}

In order to construct a stronger database from the basic $\DB$, one can use a PIR protocol $(\Q, \AA, \RR)$ (see sec.~\ref{sec:pir_bg}) and use converters on interfaces $\C$ and $\S$. The client converter is described in fig.~\ref{pir_client} and the server converter in fig.~\ref{pir_server}. Note that the client converter $\mathsf{pir}_{\mathsf{cli}}$ modifies the request $\texttt{query}$ by changing its parameter to be an index $i\in [1,n]$ corresponding to the database entry he wishes to retrieve. The server converter does the same by modifying the $\texttt{answer}$ request to have no parameter at all, the converter taking care of computing an answer using the $\AA$ algorithm of the PIR protocol. Moreover, the $\mathsf{pir}_{\mathsf{ser}}$ converter makes sure that the $\texttt{answer}$ request produces well formed answers for the clients but this converter does not alter the behavior of $\texttt{getQuery}$ and $\texttt{read}$ requests at interface $\S$. This means that those requests can still be used by the server to gather information and try to find out the index $i$ of interest for the client.

A converter can also be used on interface $\C_0$ if the database content needs to be modified in any way. For example, a converter on $\C_0$ can be used to encode the database if the PIR protocol $(\Q, \AA, \RR)$ relies on an error correcting code.

\begin{figure}[htbp]
\begin{mybox}{Converter $\mathsf{pir}_{\mathsf{cli}}$}
\vspace{-0.4cm}
  \begin{algorithm}[H]
\renewcommand{\thealgorithm}{}%no numbering
\makeatletter
\renewcommand*{\ALG@name}{Initialization}
\makeatother
\caption{Initialization}
\begin{algorithmic}%[1] for line numbers
	\State {$ind \leftarrow \perp, s\leftarrow \perp$}
\end{algorithmic}
\end{algorithm}

\vspace{-0.8cm}

\begin{algorithm}[H]
\renewcommand{\thealgorithm}{}%no numbering
\makeatletter
\renewcommand*{\ALG@name}{Interface}
\makeatother
\caption{Interface \texttt{out}}
\begin{algorithmic}%[1] for line numbers
	\Require {$(\texttt{query}, i)\in [1, n]$}
        \If {$ind = \perp$}
            \State {$s \twoheadleftarrow \mathcal{S}$}
            \State {$ind \leftarrow i$}
            \State {$q \leftarrow \Q(ind ,s)$}
            \State {$\mathbf{output}\: (\texttt{query}, q)$ at interface $\C$ of $\DB$}
        \EndIf

    \Require {$\texttt{reconstruct}$}
        \State {$\mathbf{output}\: \texttt{reconstruct}$ at interface $\C$ of $\DB$}
        \State {Let $a$ be the result}
        \If {$a \neq \perp$}
            \State {\Return $\RR(a, ind, s)$}
        \EndIf
\end{algorithmic}
\end{algorithm}

\end{mybox}
\caption{\label{pir_client} Description of the client converter $\mathsf{pir}_{\mathsf{cli}}$ for a PIR protocol $(\Q, \AA, \RR)$.}
\end{figure}

\begin{figure}[htbp]
  \begin{mybox}{Converter $\mathsf{pir}_{\mathsf{ser}}$}
    \vspace{-0.4cm}
\begin{algorithm}[H]
\renewcommand{\thealgorithm}{}%no numbering
\makeatletter
\renewcommand*{\ALG@name}{Interface}
\makeatother
\caption{Interface \texttt{out}}
\begin{algorithmic}%[1] for line numbers
    \Require {$\texttt{answer}$}
        \State {$\mathbf{output}\: \texttt{getQuery}$ at interface $\S$ of $\DB$}
        \State {Let $q$ be the result}
            \If {$q \neq \perp$}
                \State {Retrieve $\M$ with $\texttt{read}$ requests at interface $\S$}
                \State {$a \leftarrow \AA(\M, q)$}
                \State {$\mathbf{output}\: (\texttt{answer}, a)$ at interface $\S$ of $\DB$}
            \EndIf
\end{algorithmic}
\end{algorithm}

\end{mybox}
\caption{\label{pir_server} Description of the server converter $\mathsf{pir}_{\mathsf{ser}}$ for a PIR protocol $(\Q, \AA, \RR)$.}
\end{figure}

We now introduce a new database resource $\PrivDB$ with stronger security guarantees that we hope to achieve with a PIR protocol. When we talk about ``stronger security guarantees'' we mean that the capabilities of the adversary at its interface $\S$ will be weakened. In $\DB$, the adversary had access to the queries via the $\texttt{getQuery}$ request and had also access to a log file which contained the query, the answer and the possible requests sent by $\C_0$ during the initialization of the resource. In order to ensure privacy for the client's queries, the database $\PrivDB$ will hide those queries and answers from the adversary. For the PIR protocol to achieve this construction, the ideal world consisting of the private database $\PrivDB$ (together with a simulator) has to be indistinguishable from a basic database $\DB$ equipped with the PIR protocol $(\mathsf{pir}_{\mathsf{cli}}, \mathsf{pir}_{\mathsf{ser}})$ with algorithms $(\Q, \AA, \RR)$ (called the real world). This means that whatever the semi-honest adversary can do on $\PrivDB$ he could just as well do on $\DB$ equipped with the PIR protocol. Since the queries and answers are not available to the adversary in $\PrivDB$ they should be useless to him in the real world and in particular they should not help him get any information about the index $i$ desired by the user.

Furthermore, we will also require the resource $\PrivDB$ to be ``correct'' in the sense that when sending the $\mathtt{reconstruction}$ request (after a query and an answer), the user gets back the database record he queried. The resource $\PrivDB$ is given in fig.~\ref{PrivDB}.

\begin{figure}[htbp]
\begin{mybox}{Resource $\PrivDB_{\Sigma, n}$}
    \vspace{-0.4cm}
  \begin{algorithm}[H]
\renewcommand{\thealgorithm}{}%no numbering
\makeatletter
\renewcommand*{\ALG@name}{Initialization}
\makeatother
\caption{Initialization}
\begin{algorithmic}%[1] for line numbers
	\State $\textsc{Init}, \textsc{Active} \leftarrow \texttt{false}$, $\M \leftarrow [~], \Hist \leftarrow [~]$, $q \leftarrow \perp, a \leftarrow \perp, ind \leftarrow \perp$
\end{algorithmic}
\end{algorithm}

\vspace{-0.8cm}

\begin{multicols}{2}
\begin{algorithm}[H]
\renewcommand{\thealgorithm}{}%no numbering
\makeatletter
\renewcommand*{\ALG@name}{Interface}
\makeatother
\caption{Interface $\C$}
\begin{algorithmic}%[1] for line numbers
\Require {$(\texttt{query}, i)\in [1, n]$}
	\If {\Active}
	    \If {$q = \perp$}
	        \State {$q \leftarrow \texttt{ok}$}
                \State {$\Hist \leftarrow \Hist \mid\mid (\texttt{query})$}
                \State {$ind  \leftarrow i$}
	    \EndIf
        \EndIf

\Require {$\texttt{reconstruct}$}
         \If {\Active}
             \If {$a \neq \perp$}
                 \State {\Return $\M[ind]$}\footnote{If the server has the capability to overwrite data, we have to return the $ind$-th record of the database as it was when the answer was computed.}
             \EndIf
         \EndIf
\end{algorithmic}
\end{algorithm}

\columnbreak

\begin{algorithm}[H]
\renewcommand{\thealgorithm}{}%no numbering
\makeatletter
\renewcommand*{\ALG@name}{Interface}
\makeatother
\caption{Interface $\S$}
\begin{algorithmic}%[1] for line numbers
\Require {$\texttt{answer}$}
	\If {\Active}
            \If {$q \neq \perp$ \textbf{ and } $a = \perp$}
	        \State {$a \leftarrow \texttt{ok}$}
                \State {$\Hist \leftarrow \Hist \mid\mid (\texttt{answer})$}
            \EndIf
    \EndIf
\Require {$(\texttt{read}, i) \in [1, n]$}
    \State {\Return $\M[i]$}
\Require {$\texttt{getQuery}$}
    \State {\Return $q$}\footnote{The two possible values for $q$ are $\perp$ and $\texttt{ok}$.}
\Require {$\texttt{getHist}$}
    \State {\Return $\Hist$}
\end{algorithmic}
\end{algorithm}

\end{multicols}

\end{mybox}
\caption{\label{PrivDB}The private database $\PrivDB$ where queries and answers are useless to an adversary (privacy) and where the client gets the item he queried (correctness). Interface $\C_0$ remains unchanged.}
\end{figure}

The construction is valid if the resource $\DB$ with the PIR protocol is indistinguishable from the resource $\PrivDB$ with a simulator (see fig.~\ref{usmr_constr} of sec.~\ref{sec:ismr}). The simulator is here to simulate the real world behavior from the adversary point of view. The simulator is thus plugged in the $\S$ interface, processing and choosing the return values of each request of the adversary at this interface. Here, the goal of the simulator is twofold.

First, it has to simulate a real-world history with an ideal one. Recall that the ideal history doesn't have any query or answer in it and that the simulator has no way to know which index $i$ is wanted by the user. The simulation will then be as follows : the simulator will choose an index $i'$ uniformly at random and compute a query $q$ and an answer $a$ for this index using the algorithms $\Q$ and $\AA$. It will then replace the ideal query and answer of the ideal history with the simulated query and answer.

The second goal of the simulator is to maintain a simulated database if need be. Indeed, since the adversary can read the database using $(\texttt{read}, i)$ requests at interface $\S$ and since there is no protocol on the ideal resource, the simulator has to apply any transformation that the protocol could have used on the database records. For example, the simulator has to encode the ideal database content if an error correcting code is used by the real world protocol or encrypt the database if the protocol has done so in the real world. Then on a $\mathtt{read}(i)$ request to the ideal resource, the simulator returns the $i$-th record of the simulated database. The simulator is given in full detail in fig.~\ref{simPriv}.

We claim that telling apart the real view from the simulated one is as hard as attacking the privacy property of the PIR protocol. Indeed, the $\mathtt{query}$ observable behavior at interface $\C$ is the same in both worlds and the same holds for $\mathtt{answer}$. For $\mathtt{reconstruction}$ the real world has to match the ideal one where it always returns the database record queried by the user. This is the case if the PIR protocol is correct. At interface $\S$, the difference between the real and the simulated history is that the query and answer of the real one are computed using the index chosen by the client whereas in the simulated history they are computed using a random index. For a distinguisher to tell if a query and answer pair was computed using the index $i$ he chose (since it has access to the client interface) or a different index, he would have to learn some information about the index used in this execution of the PIR protocol (mainly if it was computed using $i$ or not) thus breaking the privacy property. For the $\mathtt{read}$ request available at the $\S$ interface, the simulator knows which protocol was used at the interface $\C_0$ to setup the database (with error correcting codes, encryption, \ldots) and can thus perfectly simulate the real database by applying the same transformation to the ideal database (which he has access to thanks to $\mathtt{read}$ requests).

We can conclude that if a correct and private PIR protocol is used in the construction, the advantage of the distinguisher in distinguishing the real world from the ideal world is no greater than the advantage of an attacker in telling apart a pair $(q,a)$ of the protocol's query and answer computed using an index $i$ chosen by the attacker from a pair computed using a random index.

\begin{theorem}
  Let $\Sigma$ be a finite alphabet and $n\in \N$. The protocol $\mathsf{pir} := (\mathsf{pir}_\mathsf{cli}, \mathsf{pir}_\mathsf{ser})$ described in figures \ref{pir_client} and \ref{pir_server} based on a PIR scheme $(\Q, \AA, \RR)$ constructs the private database $\PrivDB_{\Sigma, n}$ from the basic database $\DB_{\Sigma, n}$, with respect to the simulator $\mathsf{sim}^{\Q, \AA, \RR}_{Priv}$ as defined in  fig.~\ref{simPriv} and the dummy converter $\mathsf{honDB}$ (that disables any adversarial behavior). More specifically, we construct a reduction $\mathbf{C}$ such that, for all distinguishers $\mathbf{D}$,
  \[ \Delta^\mathbf{D}(\mathsf{honDB}^\S\mathsf{pir}_{\mathsf{cli}}^\C \mathsf{pir}_{\mathsf{ser}}^\S \DB_{\Sigma, n},  \mathsf{honDB}^\S\PrivDB_{\Sigma, n}) = 0  \]
  \[ \textit{and } \Delta^\mathbf{D}(\mathsf{pir}_{\mathsf{cli}}^\C \mathsf{pir}_{\mathsf{ser}}^\S \DB_{\Sigma, n},  (\mathsf{sim}^{\Q, \AA, \RR}_{Priv})^{\S}\PrivDB_{\Sigma, n}) = \Delta^{\mathbf{D}\mathbf{C}}(\mathbf{G}_0, \mathbf{G}_1)\footnote{If the PIR protocol is IT-secure, then this advantage is 0 for all distinguishers.}  \]
  where the privacy games $\mathbf{G}_b$ are described in fig.~\ref{game}. 
\end{theorem}

\begin{proof}
  We need to find a reduction $\mathbf{C}$ such that $\mathbf{C}\mathbf{G}_0 \equiv  \mathsf{pir}_{\mathsf{cli}}^\C \mathsf{pir}_{\mathsf{ser}}^\S \DB_{\Sigma, n}$ and $\mathbf{C}\mathbf{G}_1 \equiv (\mathsf{sim}^{\Q, \AA, \RR}_{Priv})^{\S}\PrivDB_{\Sigma, n}$. Then, for all distinguisher $\mathbf{D}$, we will have that the advantage of $\mathbf{D}$ in distinguishing the real system from the ideal one is
\begin{align*}
  \Delta^\mathbf{D}(\mathsf{pir}_{\mathsf{cli}}^\C \mathsf{pir}_{\mathsf{ser}}^\S \DB_{\Sigma, n},  (\mathsf{sim}^{\Q, \AA, \RR}_{Priv})^{\S}\PrivDB_{\Sigma, n}) & = \Delta^\mathbf{D} (\mathbf{C}\mathbf{G}_0, \mathbf{C}\mathbf{G}_1) \\
  & = \Delta^{\mathbf{D}\mathbf{C}}(\mathbf{G}_0, \mathbf{G}_1)
\end{align*}
which is the result we are looking for. The reduction $\mathbf{C}$ is given in fig.~\ref{reduction}.

\begin{figure}[htbp]
\begin{mybox}{Game $\mathbf{G}_b$}
\vspace{-0.4cm}
  \begin{algorithm}[H]
\renewcommand{\thealgorithm}{}%no numbering
\makeatletter
\renewcommand*{\ALG@name}{Initialization}
\makeatother
\caption{Initialization}
\begin{algorithmic}%[1] for line numbers
	\State $\M\leftarrow [~], n\leftarrow \perp$
\end{algorithmic}
\end{algorithm}

\vspace{-0.8cm}

\begin{algorithm}[H]
\renewcommand{\thealgorithm}{}%no numbering
\makeatletter
\renewcommand*{\ALG@name}{Interface}
\makeatother
\caption{Interface $\mathbf{out}$}
\begin{algorithmic}%[1] for line numbers
\Require {$(\texttt{init}, \M')$}
  \State {$\M \leftarrow \M'$}
  \State {$n\leftarrow \vert \M\vert$}
  
\Require {$(\texttt{chall}, i)$}
    \State {$s \twoheadleftarrow \mathcal{S}$}
    \If {$b = 0$}    
      \State {$j\leftarrow i$}
    \Else
      \State {$j\twoheadleftarrow [n]$}
    \EndIf
    \State {$q \leftarrow \Q(j, s)$}
    \State {$a \leftarrow \AA(\M, q)$}
    \State {\Return $(q,a)$}
\end{algorithmic}
\end{algorithm}
\end{mybox}
\caption{\label{game}The privacy games for the PIR protocol $(\Q, \AA, \RR)$.}
\end{figure}

\begin{figure}[htbp]
\begin{mybox}{Reduction $\mathbf{C}$}
\vspace{-0.4cm}
  \begin{algorithm}[H]
\renewcommand{\thealgorithm}{}%no numbering
\makeatletter
\renewcommand*{\ALG@name}{Initialization}
\makeatother
\caption{Initialization}
\begin{algorithmic}%[1] for line numbers
	\State $\textsc{Init}, \textsc{Active}, \textsc{Queried}, \textsc{Answered} \leftarrow \texttt{false}$
	\State $\M \leftarrow [~], \Hist \leftarrow [~]$, $q \leftarrow  \perp, a\leftarrow \perp, ind\leftarrow \perp$
\end{algorithmic}
\end{algorithm}

\vspace{-0.8cm}

\begin{multicols}{2}
\begin{algorithm}[H]
\renewcommand{\thealgorithm}{}%no numbering
\makeatletter
\renewcommand*{\ALG@name}{Interface}
\makeatother
\caption{Interface $\C_0$}
\begin{algorithmic}%[1] for line numbers
\Require {$(\texttt{init},  \M') \in \Sigma^n$}
         \If {\textbf{not} \Init}
             \State {$\M \leftarrow \M'$}
             \State {$\Hist \leftarrow \Hist \mid\mid (0, \mathtt{init})$}
             \State {$\Init \leftarrow \texttt{true}$}
         \EndIf
\Require {$(\texttt{read}, i) \in [1,n]$}
	\If {\Init \textbf{ and not} \Active}
	    \State {$\Hist \leftarrow \Hist \mid\mid (0, \texttt{R}, i)$}
	    \State {\Return $\M[i]$}
	\EndIf
\Require {$(\texttt{write}, i, x) \in [1,n]\times \Sigma$}
	\If {\Init \textbf{ and not} \Active}
            \State {$\Hist \leftarrow \Hist \mid\mid (0, \texttt{W}, i, x)$}
	    \State {$\M[i] \leftarrow x$}
	\EndIf
\Require {$\texttt{initComplete}$}
	\If {\Init \textbf{ and not} \Active}
      \State {$\Active \leftarrow \texttt{true}$}
      \State {Send $(\texttt{init}, \M)$ at interface $\texttt{out}$ of $\mathbf{G}_b$}
      \EndIf
\end{algorithmic}
\end{algorithm}

\columnbreak

\begin{algorithm}[H]
\renewcommand{\thealgorithm}{}%no numbering
\makeatletter
\renewcommand*{\ALG@name}{Interface}
\makeatother
\caption{Interface $\C$}
\begin{algorithmic}%[1] for line numbers
\Require {$(\texttt{query}, i)$}
  \If {\Active\textbf{ and not} \textsc{Queried}}
    \State {$\textsc{Queried} \leftarrow \texttt{true}$}
    \State {$(q, a) \leftarrow (\texttt{chall}, i)$ at interface $\texttt{out}$ of $\mathbf{G}_b$}
    \State {$ind \leftarrow i$}
    \State {$\Hist \leftarrow \Hist \mid\mid (\texttt{query}, q)$}
    \EndIf

\Require {$\texttt{reconstruct}$}
         \If {\Active\textbf{ and} \textsc{Answered}}
                 \State {\Return $\M[ind]$}
                 \EndIf
\end{algorithmic}
\end{algorithm}

\vspace{-0.8cm}

\begin{algorithm}[H]
\renewcommand{\thealgorithm}{}%no numbering
\makeatletter
\renewcommand*{\ALG@name}{Interface}
\makeatother
\caption{$\S$}
\begin{algorithmic}%[1] for line numbers
\Require {$\texttt{answer}$}
  \If {\textsc{Queried}}
    \State {$\textsc{Answered} \leftarrow \texttt{true}$}
    \State {$\Hist \leftarrow \Hist \mid\mid (\texttt{answer}, a)$}
    \EndIf
\Require {$(\texttt{read}, i) \in [1, n]$}
    \State {\Return $\M[i]$}
\Require {$\texttt{getQuery}$}
  \If {\textsc{Queried}}    
    \State {\Return $q$}
  \Else
    \State {\Return $\perp$}
    \EndIf
\Require {$\texttt{getHist}$}
  \State {\Return $\Hist$}
\end{algorithmic}
\end{algorithm}

\end{multicols}

\end{mybox}
\caption{\label{reduction}The reduction $\mathbf{C}$ connected to a privacy game $\mathbf{G}_b$.}
\end{figure}

This reduction, when connected to the game $\mathbf{G}_0$, exhibits the exact same behavior as the real system $\DB$ equipped with the protocol, and, when connected to the game $\mathbf{G}_1$, exhibits the exact same behavior as the ideal system $\PrivDB$ equipped with the simulator.
\end{proof}

Note that the proof is the same whether the PIR protocol is computationally secure or information-theoretically secure. In the first case, we consider efficient (\ie poly-time) distinguishers while in the other one we consider all possible distinguishers. In both cases, we construct the same private database $\PrivDB$, with distance (\ie best distinguishing advantage) $0$ with respect to all possible distinguishers in the IT setting and with a non-zero distance, linked to the underlying computational assumption, with respect to only efficient distinguishers in the computational setting. Since Constructive Cryptography is a composable framework, we can use this construction as a step in a construction of a stronger database or a more general system thanks to the composition theorem.

\begin{remark}
Since the random string $s$ does not leak in any way to the distinguisher or the curious server (which is accurate since it is known only to the client), the distinguisher cannot use the reconstruction algorithm $\RR$ to tell if the query and answer pair he observes really comes from the index $i$ he chose or not. 
\end{remark}

\begin{figure}[htbp]
\begin{mybox}{Simulator $\mathsf{sim}^{\Q, \AA, \RR}_{Priv}$}
\vspace{-0.4cm}
\begin{algorithm}[H]
\renewcommand{\thealgorithm}{}%no numbering
\makeatletter
\renewcommand*{\ALG@name}{Initialization}
\makeatother
\caption{Initialization}
\begin{algorithmic}%[1] for line numbers
	\State $H \leftarrow [~], \M_{sim}\leftarrow [~], pos \leftarrow 1$, $s\leftarrow \perp, ind\leftarrow \perp, q \leftarrow \perp, a\leftarrow \perp$
\end{algorithmic}
\end{algorithm}

\vspace{-0.8cm}

\begin{algorithm}[H]
\renewcommand{\thealgorithm}{}%no numbering
\makeatletter
\renewcommand*{\ALG@name}{Interface}
\makeatother
\caption{Interface $\S$}
\begin{algorithmic}%[1] for line numbers
\Require {$\texttt{answer}$}
    \State $\textsc{Update}()$
    \State \Return $a$
  
\Require {$\texttt{getHist}$}
	\State $\textsc{Update}()$
	\State \Return $H$

\Require {$\texttt{getQuery}$}
	\State $\textsc{Update}()$
	\State \Return $q$
	
\Require {$(\texttt{read}, i)$}
	\State $\textsc{Update}()$
	\State \Return $\M_{sim}[i]$
\end{algorithmic}
\end{algorithm}
\vspace{-0.6cm}
\begin{algorithmic}%[1] for line numbers
	\Procedure {\textsc{Update}}{}
            \State $\M' \leftarrow [\mathtt{read}(i)\text{ to } \PrivDB ~ \mathbf{for}~ i = 1\ldots n]$
            \State $\M_{sim} \leftarrow \mathcal{P}_{\C_0}(\M')$
            \State $\Hist \leftarrow \mathtt{getHist}\text{ to } \PrivDB$
            \For {$j = pos\ldots \vert \Hist\vert$}
                 \If {$\Hist[j] = (\mathtt{query})$}
                     \State $ind \twoheadleftarrow [1, n]$
                     \State $s \twoheadleftarrow \mathcal{S}$
                     \State $q \leftarrow \Q(ind, s)$
                     \State $H \leftarrow H\mid\mid (\mathtt{query}, q)$
                 \ElsIf {$\Hist[j] = (\mathtt{answer})$}
                     \State $a \leftarrow \AA(\M, q)$
                     \State $H \leftarrow H\mid\mid (\mathtt{answer}, a)$
                 \Else
                     \State $H \leftarrow H\mid\mid \Hist[j]$
                 \EndIf
                 \State $pos \leftarrow pos + 1$
            \EndFor
	\EndProcedure
		
\end{algorithmic}
\end{mybox}
\caption{\label{simPriv}The simulator for the protocol composed of $(\Q, \AA, \RR)$. $\mathcal{P}_{\C_0}$ denotes the transformation applied to the memory $\M$ by a possible converter plugged in interface $\C_0$.}
\end{figure}

The efficiency of the PIR protocol is given by : $\vert q\vert + \vert a\vert$ for its communication complexity, and its computational complexity is given by the sum of those of $\Q$, $\AA$ and $\RR$.

\subsection{Generalization to multiple servers}

We present $\MultDB$, a generalization of $\DB$ in the multiple server case. For the sake of simplicity, we put the same database on all servers but it is possible to modify the behavior of interface $\C_0$ to allow the client to modify each database independently. Similarly, we generalize the converters into $\mathsf{mult\_pir}_{\mathsf{cli}}$ and $\mathsf{mult\_pir}_{\mathsf{ser}_j}$ in the multiple server setting. Those converters are described in figures \ref{multpir_client} and \ref{multpir_server}. The three algorithms $\Q, \AA$ and $\RR$ are still present and used in the same way as before.

In the multi-server case, we introduce a threshold $t$ which allows the curious servers to form a coalition of size at most $t$ in order to share the capabilities of their interfaces. This means that each server in the coalition can see the histories of the other members as well as read their memory. The database $\MultDB_{n, k, t}$ can be configured by choosing $n$ the size of the database, $k$ the number of servers holding this database (or a modified/distributed version of it) and $t$ the largest number of malicious servers allowed to cooperate. This resource is described in fig.~\ref{MultDB}. 

We can then define a new resource with better security guarantees than those of $\MultDB$. Just like in the one server case, the new resource $\PrivMultDB$ requires that $\texttt{reconstruct}$ returns the database record desired by the client when he issued its $\texttt{query}$ request. $\PrivMultDB$ also requires that the curious servers $\S_{j}$ do not get to see the queries and answers when using $\texttt{getHist}$. This resource is described in fig.~\ref{PrivMultDB}. For the construction of $\PrivMultDB_{n,k,t}$ to hold, we require the PIR protocol used to be both correct and $t$-private.

\begin{figure}[htbp]
  \begin{mybox}{Resource $\MultDB_{\Sigma, n, k, t}$}
    \vspace{-0.4cm}
  \begin{algorithm}[H]
\renewcommand{\thealgorithm}{}%no numbering
\makeatletter
\renewcommand*{\ALG@name}{Initialization}
\makeatother
\caption{Initialization}
\begin{algorithmic}%[1] for line numbers
	\State $\textsc{Init}, \textsc{Active}, \textsc{Coalition} \leftarrow \texttt{false}$
	\State $\Hist_j \leftarrow [~], q_j \leftarrow  \perp, a_j\leftarrow \perp, b_j \leftarrow \texttt{false } \mathbf{for} ~ j = 1\ldots k$
	\State $\M \leftarrow [~]$
\end{algorithmic}
\end{algorithm}

\vspace{-0.8cm}

\begin{multicols}{2}
\begin{algorithm}[H]
\renewcommand{\thealgorithm}{}%no numbering
\makeatletter
\renewcommand*{\ALG@name}{Interface}
\makeatother
\caption{Interface $\C_0$}
\begin{algorithmic}%[1] for line numbers
\Require {$(\texttt{init}, \M') \in \Sigma^n$}
         \If {\textbf{not} \Init}
             \State {$\M \leftarrow \M'$}
             \For {$j = 1\ldots k$}
                  \State {$\Hist_j \leftarrow \Hist_j \mid\mid (0, \mathtt{init})$}
             \EndFor
             \State {$\Init \leftarrow \texttt{true}$}
         \EndIf

\Require {$(\texttt{read}, i) \in [1,n]$}
	\If {\Init \textbf{ and not} \Active}
            \For {$j = 1\ldots k$}
	         \State {$\Hist_j \leftarrow \Hist_j \mid\mid (0, \texttt{R}, i)$}
            \EndFor
	    \State {\Return $\M[i]$}
	\EndIf

\Require {$(\texttt{write}, i, x) \in [1,n]\times \Sigma$}
	\If {\Init \textbf{ and not} \Active}
            \For {$j = 1\ldots k$}
                 \State {$\Hist_j \leftarrow \Hist_j \mid\mid (0, \texttt{W}, i, x)$}
            \EndFor
	    \State {$\M[i] \leftarrow x$}
	\EndIf

\Require {$\texttt{initComplete}$}
	\If {\Init \textbf{ and not} \Active}
            \State {$\Active \leftarrow \texttt{true}$}
	\EndIf	
\end{algorithmic}
\end{algorithm}

\columnbreak

\begin{algorithm}[H]
\renewcommand{\thealgorithm}{}%no numbering
\makeatletter
\renewcommand*{\ALG@name}{Interface}
\makeatother
\caption{Interface $\C$}
\begin{algorithmic}%[1] for line numbers
\Require {$(\texttt{query}, q'_1, \ldots, q'_k)$}
	\If {\Active}
	    \If {$(q_1, \ldots, q_k) = (\perp, \ldots, \perp)$}
            \For {$j = 1\ldots k$}
                 \State {$q_j \leftarrow q'_j$}
                     \State {$\Hist_j \leftarrow \Hist_j \mid\mid (\texttt{query}, j,  q_j)$}
                \EndFor
	    \EndIf
        \EndIf

\Require {$\texttt{reconstruct}$}
         \If {\Active}
             \If {$a_1 \neq \perp \mathbf{and} ~\ldots ~ \mathbf{and} ~a_k \neq \perp$}
                 \State {\Return $(a_1, \ldots, a_k)$}
             \EndIf
         \EndIf
\end{algorithmic}
\end{algorithm}

\vspace{-0.8cm}

\begin{algorithm}[H]
\renewcommand{\thealgorithm}{}%no numbering
\makeatletter
\renewcommand*{\ALG@name}{Interface}
\makeatother
\caption{Interface $\W$}
\begin{algorithmic}%[1] for line numbers
	\Require {\small{$(\texttt{formCoalition}, b'_1, \ldots, b'_k)~ \in~ \lbrace \texttt{true}$, $\texttt{false}\rbrace^k$}}
                 \If {$\vert \lbrace 1\leq j\leq k \mid b'_j = \mathtt{true}\rbrace\vert \leq t$ \textbf{ and not } $\textsc{Coalition}$}
                     \State {$(b_1, \ldots, b_k) \leftarrow (b'_1, \ldots, b'_k)$}
                     \State {$\textsc{Coalition} \leftarrow \texttt{true}$}
                 \EndIf
\end{algorithmic}
\end{algorithm}

\end{multicols}

\vspace{-0.8cm}

\begin{algorithm}[H]
\renewcommand{\thealgorithm}{}%no numbering
\makeatletter
\renewcommand*{\ALG@name}{Interface}
\makeatother
\caption{Interface $\S_{j}, j\in [1,k]$}
\begin{algorithmic}%[1] for line numbers
\Require {$(\texttt{answer}, a)$}
	\If {\Active}
        \If {$q_j \neq \perp$ \textbf{ and } $a_j = \perp$}
                \State {$a_j \leftarrow a$}
                \State {$\Hist_j \leftarrow \Hist_j \mid\mid (\texttt{answer}, j, a_j)$}
            \EndIf
            \EndIf
\Require{$(\texttt{read}, i)\in [1, n]$}
    \State {\Return $\M[i]$}
\Require {$\texttt{getHist}$}
  \If {\textsc{Coalition} \textbf{ and } $b_j$}
      \State {\Return $\Hist_j$}
  \EndIf
\Require {$\texttt{getQuery}$}
    \If {\textsc{Coalition} \textbf{ and } $b_j$}
        \State {\Return $q_j$}
    \EndIf
\end{algorithmic}
\end{algorithm}

\end{mybox}
\caption{\label{MultDB}The basic database $\MultDB$ for $k$ servers where at most $t$ of them can form a coalition, allowing them to share information together.}
\end{figure}

\begin{figure}[htbp]
\begin{mybox}{Converter $\mathsf{mult\_pir}_{\mathsf{cli}}$}
    \vspace{-0.4cm}
  \begin{algorithm}[H]
\renewcommand{\thealgorithm}{}%no numbering
\makeatletter
\renewcommand*{\ALG@name}{Initialization}
\makeatother
\caption{Initialization}
\begin{algorithmic}%[1] for line numbers
	\State {$ind \leftarrow \perp, s\leftarrow \perp$}
\end{algorithmic}
\end{algorithm}

\vspace{-0.8cm}

\begin{algorithm}[H]
\renewcommand{\thealgorithm}{}%no numbering
\makeatletter
\renewcommand*{\ALG@name}{Interface}
\makeatother
\caption{Interface \texttt{out}}
\begin{algorithmic}%[1] for line numbers
	\Require {$(\texttt{query}, i)\in [1, n]$}
        \If {$ind = \perp$}
            \State {$s \twoheadleftarrow \mathcal{S}$}
            \State {$ind \leftarrow i$}
            \State {$(q_1, \ldots, q_k) \leftarrow \Q(ind ,s)$}
            \State {$\mathbf{output}\: (\texttt{query}, q_1, \ldots, q_k)$ at interface $\C$ of $\MultDB$}
        \EndIf
        \State

    \Require {$\texttt{reconstruct}$}
        \State {$\mathbf{output}\: \texttt{reconstruct}$ at interface $\C$ of $\MultDB$}
        \State {Let $(a_1, \ldots, a_k)$ be the result}
        \If {$a_1 \neq \perp$ \textbf{ and } $a_k \neq \perp$}
            \State {\Return $\RR(a_1, \ldots, a_k, ind, s)$}
        \EndIf
\end{algorithmic}
\end{algorithm}

\end{mybox}
\caption{\label{multpir_client} Description of the client converter $\mathsf{mult\_pir}_{\mathsf{cli}}$ for a PIR protocol $(\Q, \AA, \RR)$ in the multi-server setting.}
\end{figure}

\begin{figure}[htbp]
\begin{mybox}{Converter $\mathsf{mult\_pir}_{\mathsf{ser}_j}$}
    \vspace{-0.4cm}
  \begin{algorithm}[H]
\renewcommand{\thealgorithm}{}%no numbering
\makeatletter
\renewcommand*{\ALG@name}{Interface}
\makeatother
\caption{Interface \texttt{out}}
\begin{algorithmic}%[1] for line numbers
    \Require {$\texttt{answer}$}
        \State {$\mathbf{output}\: \texttt{getQuery}$ at interface $\S_{j}$ of $\MultDB$}
        \State {Let $q_j$ be the result}
            \If {$q_j \neq \perp$}
                \State {Retrieve $\M$ with $\texttt{read}$ requests at interface $\S_{j}$}
                \State {$a_j \leftarrow \AA(j, \M, q_j)$}
                \State {$\mathbf{output}\: (\texttt{answer}, a_j)$ at interface $\S_{j}$ of $\MultDB$}
            \EndIf    
\end{algorithmic}
\end{algorithm}

\end{mybox}
\caption{\label{multpir_server} Description of the server converter $\mathsf{mult\_pir}_{\mathsf{ser}_j}$ for a PIR protocol $(\Q, \AA, \RR)$ in the multi-server setting.}
\end{figure}

\begin{figure}[htbp]
\begin{mybox}{Resource $\PrivMultDB_{\Sigma, n, k, t}$}
    \vspace{-0.4cm}
  \begin{algorithm}[H]
\renewcommand{\thealgorithm}{}%no numbering
\makeatletter
\renewcommand*{\ALG@name}{Initialization}
\makeatother
\caption{Initialization}
\begin{algorithmic}%[1] for line numbers
	\State $\textsc{Init}, \textsc{Active}, \textsc{Coalition} \leftarrow \texttt{false}$
	\State $\Hist_j \leftarrow [~], q_j \leftarrow  \perp, a_j\leftarrow \perp, b_j\leftarrow \texttt{false } \mathbf{for} ~ j = 1\ldots k$
	\State $\M \leftarrow [~], ind  \leftarrow \perp$
\end{algorithmic}
\end{algorithm}

\vspace{-0.8cm}

\begin{multicols}{2}

\begin{algorithm}[H]
\renewcommand{\thealgorithm}{}%no numbering
\makeatletter
\renewcommand*{\ALG@name}{Interface}
\makeatother
\caption{Interface $\C$}
\begin{algorithmic}%[1] for line numbers
\Require {$(\texttt{query}, i)\in [1, n]$}
	\If {\Active}
	    \If {$(q_1,  \ldots, q_k) = (\perp, \ldots, \perp)$}
	        \State {$(q_1, \ldots, q_k) \leftarrow (\mathtt{ok}, \ldots, \mathtt{ok})$}
                \For {$j = 1\ldots k$}
                     \State {$\Hist_j \leftarrow \Hist_j \mid\mid (\texttt{query}, j)$}
                \EndFor
                \State {$ind  \leftarrow i$}
	    \EndIf
        \EndIf
        \State

\Require {$\texttt{reconstruct}$}
         \If {\Active}
             \If {$a_1 \neq \perp \mathbf{and} ~\ldots ~ \mathbf{and} ~a_k \neq \perp$}
                 \State {\Return $\M[ind]$}
             \EndIf
         \EndIf
\end{algorithmic}
\end{algorithm}

\columnbreak

\vspace{-0.8cm}

\begin{algorithm}[H]
\renewcommand{\thealgorithm}{}%no numbering
\makeatletter
\renewcommand*{\ALG@name}{Interface}
\makeatother
\caption{Interface $\S_{j}, j\in [1,k]$}
\begin{algorithmic}%[1] for line numbers
\Require {$\texttt{answer}$}
	\If {\Active}
            \If {$q_j \neq \perp$ \textbf{ and } $a_j = \perp$}
                \State {$a_j \leftarrow \texttt{ok}$}
                \State {$\Hist_j \leftarrow \Hist_j \mid\mid (\texttt{answer}, j)$}
            \EndIf
            \EndIf
            \State
              \Require {$\texttt{getHist}$}
  \If {\textsc{Coalition} \textbf{ and } $b_j$}
      \State {\Return $\Hist_j$}
  \EndIf
  \State

  \Require {$(\texttt{read}, i) \in [1, n]$}
  \If {\textsc{Coalition} \textbf{ and } $b_j$}
      \State {\Return $\M[i]$}
  \EndIf

\end{algorithmic}
\end{algorithm}
\end{multicols}

\vspace{-0.8cm}

\begin{algorithm}[H]
\renewcommand{\thealgorithm}{}%no numbering
\makeatletter
\renewcommand*{\ALG@name}{Interface}
\makeatother
\caption{Interface $\W$}
\begin{algorithmic}%[1] for line numbers
	\Require {$(\texttt{formCoalition}, b'_1, \ldots, b'_k)~ \in~ \lbrace \texttt{true}, \texttt{false}\rbrace^k$}
                 \If {$\vert \lbrace 1\leq j\leq k \mid b'_j = \mathtt{true}\rbrace\vert \leq t$ \textbf{ and not } $\textsc{Coalition}$}
                     \State {$(b_1, \ldots, b_k) \leftarrow (b'_1, \ldots, b'_k)$}
                     \State {$\textsc{Coalition} \leftarrow \texttt{true}$}
                 \EndIf
\end{algorithmic}
\end{algorithm}
\end{mybox}
\caption{\label{PrivMultDB}The private (and correct) database $\PrivMultDB$ with $k$ servers where at most $t$ of them can form a coalition. The interface $\C_0$ remains unchanged.}
\end{figure}

\subsection{Instantiations in the multi-server case}

We give two PIR schemes that construct $\PrivMultDB$ from $\MultDB$. The first scheme is based on Locally Decodable Codes (LDC) introduced by Katz \etal \cite{katz00}. The second uses Shamir's secret sharing \cite{S79}. The focus of this section is not to explain LDCs or secret sharing but rather to showcase the power and the flexibility of our construction.

\subsubsection{Using Locally Decodable Codes}

We give a short and informal introduction to LDCs.

\begin{definition}[Locally Decodable Code - informal]
We say that a code $C$ of length $n$ is locally decodable with locality $k < n$ if there exists a randomized decoding algorithm which, given a codeword $C(x)$ and an integer $i$, is capable of recovering $x_i$ by reading at most $k$ coordinates of $C(x)$. This must hold even if $C(x)$ is corrupted on a small enough fraction of coordinates.
\end{definition}
We will also need to define the following property of locally decodable codes.

\begin{definition}[Smoothness - informal]
We say that a LDC $C$ with locality $k$ is $t$-smooth if for every index $i$ and every query $q := (q_1, \ldots, q_k)$ issued by the decoding algorithm, any restriction of $q$ to at most $t$ of its coordinates are uniformly distributed.
\end{definition}

We can now describe the following PIR protocol. Let $C$ be a $t$-smooth LDC with locality $k$, dimension $n$ and length $n'$. We need a converter on interface $\C_0$ that encodes the database $\M$ and distributes it on the $k$ different servers of $\MultDB_{n', k, t}$. Let $\M'$ be the encoded database. Let $s$ be a random string that will be used as a source of randomness in algorithm $\Q$. On input $i$ and $s$, the algorithm $\Q$ chooses $k$ queries $q_1, \ldots, q_k$ which corresponds to $k$ indices to be read in the codeword. Each one is sent to the corresponding server. The answer algorithm returns $\A(j, \M', q_j) := \M'[q_j]$. The reconstruction algorithm $\R$ runs the local decoder on inputs $a_1, \ldots, a_k$ and $i$.

We claim that this protocol used on $\MultDB_{n', k, t}$ securely constructs, for all distinguishers, the private database $\PrivMultDB_{n, k, t}$. Indeed, since the code is $t$-smooth, the distributions of the restrictions of the random variables $\Q(i, \cdot)_j$ to at most $t$ queries are identical for all $i\in [n]$ since they are in fact all uniform. This means that a coalition of at most $t$ servers gain no information on $i$ by sharing their queries together. It is thus impossible even for an unbounded distinguisher to distinguish between real queries and simulated ones if the coalition size is at most $t$.

% TODO définir i et j puis 1) les queries sont uniformes (smoothness) 2) elles sont donc indentiques pour le distinguisher

\subsubsection{Using Secret-Sharing}

Let $k$ be the number of servers, $t$ the desired coalition size threshold, $n$ the size of the database and $i$ the index of the database record desired by the client. We present a $t$-private PIR scheme using the $(t, k)$-Shamir secret-sharing \cite{S79}. The database $\M$ is viewed as a matrix with $c$ columns where $c$ is chosen to be the closest possible integer to $\sqrt{n}$.

The client starts by choosing $k$ non zero evaluation points $\alpha_1, \ldots, \alpha_k$ in a suitable finite field. On inputs $i$ and $s$, the algorithm $\Q$ determines the column $c_i$ where the $i$-th record lies, and secret shares component-wise the vector $e_{c_i}$ (with a $1$ in position $c_i$ and $0$ everywhere else) into $k$ shares $s_1, \ldots, s_k$ using Shamir secret-sharing. The shares $s_1, \ldots, s_k$ are the output of the algorithm $\Q$. The evaluation points and the source of randomness needed for these operations are given in the string $s$.

On inputs $\M$, $j$, and $s_j$ the algorithm $\A$ computes the product $a_j := \M s_j$ and outputs it.

Finally, on inputs $\M, i, a_1, \ldots, a_k, s$ the algorithm $\R$ can compute the Lagrange interpolation for each component of the $a_j$ to reconstruct the secret :

\begin{align*}
  & \sum_{j = 1}^{k}{a_j\prod_{i = 1, i\neq j}{\frac{\alpha_i}{\alpha_i - \alpha_j}}} \\
  = & \sum_{j = 1}^{k}{\M s_j\prod_{i = 1, i\neq j}{\frac{\alpha_i}{\alpha_i - \alpha_j}}} \\
  = & \M e_{c_i}
\end{align*}

The algorithm thus recovers the $c_i$-th column of the database which contains the $i$-th record, the one of interest for the client. $\R$ returns this record.

Here, the correctness of the PIR protocol directly follows from the correctness of the secret sharing scheme. The same holds for $t$-privacy: if a coalition of at most $t$ servers can't learn anything about $e_{c_i}$ from their combined shares then they can't learn anything about the index $i$ requested by the client. This protocol thus yields the same construction guarantees as the one described using a PIR protocol based on locally decodable codes.

\subsection{The case of Byzantine servers}

We can further strengthen the capabilities of $\PrivMultDB$ by allowing the servers to send an ill-formed answer or even to not answer at all after receiving the client's query. Such a server is called a Byzantine server. In the following, we introduce a threshold $u$ which is an upper-bound on the number of Byzantine servers. After receiving the client's query, a Byzantine server can choose to assign the special symbol $\epsilon$ to its answer. This symbol means that the answer is of no use to the client (corresponding to an absence of answer or an ill-formed answer during the execution of the PIR protocol). The Byzantine servers are designated by the environment (at interface $\W$). We present the new $\ByzPrivMultDB_{\Sigma, n, k, t, u}$ resource in fig.~\ref{ByzPrivMultDB}. We also need to allow at most $u$ Byzantine servers in the real resource $\MultDB$, which defines a new resource $\ByzMultDB_{\Sigma, n, k, u}$. The additions needed being the same as the ones made to $\PrivMultDB$, we also refer to fig.~\ref{ByzPrivMultDB} for those. The converters $\mathsf{mult\_pir}_{\mathsf{cli}}$ and $\mathsf{mult\_pir}_{\mathsf{ser}_j}$ remain unchanged.

The simulator used in the proof of construction of the new resource $\ByzPrivMultDB$ needs to be slightly modified to account for the aforementioned additions. In particular, the simulator just forwards the requests $\texttt{badAnswer}$ to interface $\S_{j}$ and when it simulates the history of the $j$-th server, if an entry of the form $(\texttt{answer}, j, \epsilon)$ is present, the simulator does not need to simulate an answer and can just copy this entry in its simulated history for the $j$-th server.

Finally, we can relax the definition of correctness for PIR protocols using the aforementioned threshold $u$ by saying that a PIR protocol is $u$-correct if the user can recover its desired record even if at most $u$ servers deviate from the protocol by providing incorrect answers or no answers at all.

\begin{figure}[htbp]
\begin{mybox}{Resource $\ByzPrivMultDB_{\Sigma, n, k, t, u}$}
    \vspace{-0.4cm}
  \begin{algorithm}[H]
\renewcommand{\thealgorithm}{}%no numbering
\makeatletter
\renewcommand*{\ALG@name}{Initialization}
\makeatother
\caption{Initialization}
\begin{algorithmic}%[1] for line numbers
	\State $\textsc{Byzantines} \leftarrow \texttt{false}$
	\State $c_j\leftarrow \texttt{false } \mathbf{for} ~ j = 1\ldots k$
\end{algorithmic}
\end{algorithm}

\vspace{-0.8cm}

\begin{algorithm}[H]
\renewcommand{\thealgorithm}{}%no numbering
\makeatletter
\renewcommand*{\ALG@name}{Interface}
\makeatother
\caption{Interface $\S_{j}, j\in [1,k]$}
\begin{algorithmic}%[1] for line numbers
\Require {$\texttt{badAnswer}$}
\If {\textsc{Byzantines} \textbf{ and } $c_j$ \textbf{ and } $a_j = \perp$}
  \If {$q_j \neq \perp$}
    \State {$a_j \leftarrow \epsilon$}
    \State {$\textsc{Hist}_j \leftarrow \textsc{Hist}_j \mid\mid (\texttt{answer}, j, \epsilon)$}
  \EndIf
\EndIf
      
\end{algorithmic}
\end{algorithm}

\vspace{-0.8cm}

\begin{algorithm}[H]
\renewcommand{\thealgorithm}{}%no numbering
\makeatletter
\renewcommand*{\ALG@name}{Interface}
\makeatother
\caption{Interface $\W$}
\begin{algorithmic}%[1] for line numbers
\Require {$(\texttt{formByzantines}, c'_1, \ldots, c'_k)~ \in~ \lbrace \texttt{true}, \texttt{false}\rbrace^k$}
                 \If {$\vert \lbrace 1\leq j\leq k \mid c'_j = \mathtt{true}\rbrace\vert \leq u$ \textbf{ and not } $\textsc{Byzantines}$}
                     \State {$(c_1, \ldots, c_k) \leftarrow (c'_1, \ldots, c'_k)$}
                     \State {$\textsc{Byzantines} \leftarrow \texttt{true}$}
                 \EndIf                     
\end{algorithmic}
\end{algorithm}

\end{mybox}
\caption{\label{ByzPrivMultDB}The private and Byzantine-resistant database $\ByzPrivMultDB$ with $k$ servers where at most $t$ of them can form a coalition to share information and $u$ of them can endorse a Byzantine behavior. Only the additions to $\PrivMultDB$ appear. The same additions are made to $\MultDB$ to define $\ByzMultDB$.}
\end{figure}

\section{Future work}
\label{sec:ccl}
A natural further study would be to include more interactive protocols as well as more PIR variants (e.g. SPIR, batch PIR) to our construction. For SPIR, we use the flexibility of CC to give security guarantees to the server. This is done using specifications intersection \cite{MPR21}. For UE, it would be interesting to extend our theorem \ref{thUE} to the $\mathsf{CCA}$ context of \cite{FMM21}.

\bibliography{bibliography}

%\begin{subappendices}
%\renewcommand{\thesection}{\Alph{section}}%
% or try \arabic{section}
\appendix

\section{Security Games}
\label{sec:orgfb11b76}
Fig.~\ref{oracles} is a description, taken from \cite{boyd20_fast_secur_updat_encry}, of the oracles used in the $\CPA$ setting.

   \begin{figure}[htbp]
\begin{multicols}{2}
 \begin{algorithm}[H]
 \renewcommand{\thealgorithm}{}%no numbering
 \makeatletter
 \renewcommand*{\ALG@name}{}
 \makeatother
 \caption{$\mathsf{Setup}(1^\lambda)$}
 \begin{algorithmic}%[1] for line numbers
	 \State {$k_0 \leftarrow \UE.\KG(1^\lambda)$}
     \State {$\Delta_0 \leftarrow \perp, e \leftarrow 0$}
     \State {$\mathcal{L} \leftarrow [~], \textsc{Chall} \leftarrow \texttt{false}$}
 \end{algorithmic}
 \end{algorithm}

  \vspace{-0.8cm}
 
  \begin{algorithm}[H]
 \renewcommand{\thealgorithm}{}%no numbering
 \makeatletter
 \renewcommand*{\ALG@name}{}
 \makeatother
 \caption{$\mathcal{O}.\Enc(m)$}
 \begin{algorithmic}%[1] for line numbers
	 \State {$c \leftarrow \UE.\Enc_{k_e}(m)$}
     \State {$\mathcal{L} \leftarrow \mathcal{L} \mid\mid (c, e, m)$}
     \State {$\Return\: c$}
 \end{algorithmic}
 \end{algorithm}

  \vspace{-0.8cm}
 
  \begin{algorithm}[H]
  \renewcommand{\thealgorithm}{}%no numbering
 \makeatletter
 \renewcommand*{\ALG@name}{}
 \makeatother
 \caption{$\mathcal{O}.\mathsf{Next}()$}
 \begin{algorithmic}%[1] for line numbers
     \State {$e \leftarrow e + 1$}
	 \State {$k_e \leftarrow \UE.\KG(1^\lambda)$}
     \State {$\Delta_e \leftarrow \UE.\TG(k_{e-1}, k_e)$}
     \If {\textsc{Chall}}
       \State {$\tilde{c}_e \leftarrow \UE.\Upd_{\Delta_e}(\tilde{c}_{e-1})$}
     \EndIf
 \end{algorithmic}
 \end{algorithm}

   \vspace{-0.8cm}
 
  \begin{algorithm}[H]
 \renewcommand{\thealgorithm}{}%no numbering
 \makeatletter
 \renewcommand*{\ALG@name}{}
 \makeatother
 \caption{$\mathcal{O}.\Upd(c_{e-1})$}
 \begin{algorithmic}%[1] for line numbers
     \If {$(c_{e-1}, e-1, m)\notin \mathcal{L}$}
       \State {$\Return \perp$}
     \EndIf
	 \State {$c_e \leftarrow \UE.\Upd_{\Delta_e}(c_{e-1})$}
     \State {$\mathcal{L} \leftarrow \mathcal{L} \mid\mid (c_e, e, m)$}
     \State {$\Return\: c_e$}
 \end{algorithmic}
  \end{algorithm}
  
 \columnbreak

  \begin{algorithm}[H]
 \renewcommand{\thealgorithm}{}%no numbering
 \makeatletter
 \renewcommand*{\ALG@name}{}
 \makeatother
 \caption{$\mathcal{O}.\mathsf{Corr}(\mathsf{inp}, \hat{e})$}
 \begin{algorithmic}%[1] for line numbers
     \If {$\hat{e} > e$}
       \State {\Return $\perp$}
     \EndIf
     \If {$\mathsf{inp} = \mathsf{key}$}
	   \State {\Return $k_{\hat{e}}$}
     \EndIf
     \If {$\mathsf{inp} = \mathsf{token}$}
	   \State {\Return $\Delta_{\hat{e}}$}
     \EndIf
 \end{algorithmic}
  \end{algorithm}

  \vspace{-0.8cm}

    \begin{algorithm}[H]
 \renewcommand{\thealgorithm}{}%no numbering
 \makeatletter
 \renewcommand*{\ALG@name}{}
 \makeatother
 \caption{$\mathcal{O}.\mathsf{Chall}(\bar{m}, \bar{c})$}
 \begin{algorithmic}%[1] for line numbers
     \If {\textsc{Chall}}
       \State {\Return $\perp$}
     \EndIf
     \State {$\textsc{Chall} \leftarrow \texttt{true}$}
     \If {$(\bar{c}, e - 1, \cdot)\notin \mathcal{L}$}
	   \State {\Return $\perp$}
     \EndIf
     \If {$b = 0$}
	   \State {$\tilde{c}_e \leftarrow \UE.\Enc_{k_e}(\bar{m})$}
     \Else
       \State {$\tilde{c}_e \leftarrow \UE.\Upd_{\Delta_e}(\bar{c})$}
     \EndIf
     \State {\Return $\tilde{c}_e$}
 \end{algorithmic}
  \end{algorithm}

  \vspace{-0.8cm}

  \begin{algorithm}[H]
 \renewcommand{\thealgorithm}{}%no numbering
 \makeatletter
 \renewcommand*{\ALG@name}{}
 \makeatother
 \caption{$\mathcal{O}.\mathsf{UpdC}$}
 \begin{algorithmic}%[1] for line numbers
     \If {\textbf{not} \textsc{Chall}}
       \State {\Return $\perp$}
     \EndIf
     \State {\Return $\tilde{c}_e$}
 \end{algorithmic}
  \end{algorithm}  
\end{multicols}
\caption{\label{oracles}Description of the oracles used in UE security games.}
\end{figure}

\section{$\mathsf{IND}\text{-}\mathsf{ENC}\text{+}\mathsf{UPD}\text{-}\mathsf{CPA}$ security is sufficient for constructing $\cUSMR^+$}
\label{sec:org41788ad}

\begin{proof}
  Let \(\R := \enc^\C \upd^\S [\USMR^+, \UpdKey]\) be the the real system and \(\I := \sigma_{k, \CPAp}^\S \cUSMR^+\) be the ideal system. In order to determine the advantage of a distinguisher in distinguishing \(\R\) from \(\I\), denoted by \(\Delta^\D(\R, \I)\), we proceed with a sequence of systems. We introduce a hybrid system \(\SS\), then we determine the distinguishing advantages \(\Delta^\D(\R, \SS)\) and \(\Delta^\D(\SS, \I)\), the triangular inequality allows us to bound \(\Delta^\D(\R, \I)\) by the sum of those two advantages.
\medskip

\noindent
\textbullet \(\:\)  Let \(\SS\) be a resource that behaves just like \(\R\) when leaking updated ciphertexts and just like \(\I\) when leaking fresh encryptions. Concretely, \(\SS\) maintains the same database as \(\R\) using the UE scheme and, when \(\SS\) is asked to leak an updated ciphertext it returns it as is but when \(\SS\) is asked to leak a fresh ciphertext encrypting $\M[k]$, it returns an encryption of a random \(\bar{x}\) of length $\vert \M[k]\vert$.

Let \(q\) be an upper bound on the number of $(\texttt{write}, k, .)$ queries issued to the systems. We define a hybrid resource \(\H_i\) that behaves just like \(\R\) on the first \(i\) ($\texttt{write}, k, .$) queries and like \(\SS\) afterwards. Then we define a reduction \(\CC_i\) that behaves like \(\H_i\) except it uses the game \(\G^{\ENCCPA}_b\) oracles instead of doing the UE operations by itself and on the \(i\)-th $(\texttt{write}, k, .)$ request (of the form \((\texttt{write}, k, x)\)) it challenges the game with input \((x, \bar{x})\) to receive the ciphertext. We have

\[ \R \equiv \H_q \text{ and } \SS \equiv \H_0 \]

and

\[ \H_i \equiv \G_0^\ENCCPA \CC_i \equiv \G_1^\ENCCPA \CC_{i+1} \]

Indeed, this can be seen on the following timeline (\ref{tab:org71ef67a}) 

\begin{table}[htbp]
\centering
\begin{tabular}{|c|c|c|c|c|}
\hline
\(j\)-th $(\texttt{write}, k, .)$ query & \(j<i\) & \(j=i\) & \(j=i+1\) & \(j>i+1\)\\
\hline
\({\color{green}\G_0^\ENCCPA} \CC_i\) & \(\Enc(x)\) & \color{green}{$\Enc(x)$} & \(\Enc(\bar{x})\) & \(\Enc(\bar{x})\)\\
\hline
\({\color{green}\G_1^\ENCCPA} \CC_{i+1}\) & \(\Enc(x)\) & \(\Enc(x)\) & \color{green}{$\Enc(\bar{x})$} & \(\Enc(\bar{x})\)\\
\hline
\end{tabular}
\caption{\label{tab:org71ef67a}Leakage behavior of both systems for each \((\texttt{write}, k, .)\) request.}

\end{table}

Let \(\CC_I\) be a reduction that samples \(i\in [1, q]\) at random and behaves like \(\CC_i\) and define \(\D' := \D\CC_I\). We have,

\[ \Pr[\D'(\G_0^\ENCCPA) = 1] = \frac{1}{q} \sum_{i = 1}^q{\Pr[\D(\CC_i \G_0^\ENCCPA) = 1]} \]

and

\begin{align*}
\Pr[\D'(\G_1^\ENCCPA) = 1] & = \frac{1}{q} \sum_{i=1}^q{\Pr[\D(\CC_i \G_1^\ENCCPA) = 1]} \\
& = \frac{1}{q} \sum_{i=0}^{q-1}{\Pr[\D(\CC_i \G_0^\ENCCPA) = 1]}
\end{align*}

Finally, the advantage of the distinguisher in distinguishing system \(\R\) from \(\SS\) is

\begin{align*}
  \Delta^\D (\R, \SS) & = \Delta^\D(\H_q, \H_0) \\
                      & = \Delta^\D(\CC_q \G_0^\ENCCPA, \CC_0 \G_0^\ENCCPA) \\
                      & = \vert \Pr[\D(\CC_q \G_0^\ENCCPA) = 1] - \Pr[\D(\CC_0 \G_0^\ENCCPA) = 1] \vert \\
                      & = \vert \sum_{i=1}^q{\Pr[\D(\CC_i \G_0^\ENCCPA) = 1]} - \sum_{i=0}^{q-1}{\Pr[\D(\CC_i \G_0^\ENCCPA) = 1]} \vert \\
                      & = q \cdot \vert \Pr[\D'(\G_0^\ENCCPA) = 1] - \Pr[\D'(\G_1^\ENCCPA) = 1] \vert \\
                      & = q\cdot \Delta^{\D'}(\G_0^\ENCCPA, \G_1^\ENCCPA)
\end{align*}

\medskip

\noindent
\textbullet \(\:\) Let us consider the systems \(\SS\) and \(\I\). By definition, \(\SS\) behaves just like \(\I\) when leaking fresh encryptions but, when asked to leak what should be an updated ciphertext encrypting $\M[k]$, \(\SS\) returns this updated ciphertext while \(\I\) simply returns an update of an encryption of a random \(\bar{x}\) of length $\vert \M[k]\vert$.

Let \(r\) be an upper bound on the number of update queries issued to the systems. We define a hybrid resource \(\H'_i\) that behaves just like \(\SS\) on the first \(i\) update queries to location $k$ and like \(\I\) afterwards. Then we define a reduction \(\CC'_i\) that behaves like \(\H'_i\) except it uses the game \(\G^{\UPDCPA}_b\) oracles instead of doing the UE operations by itself and on the \(i\)-th update computation for the encryption $c$ of $\M[k]$, it challenges the game with input \((c, \bar{c})\) to receive either an updated version of \(c\) or \(\bar{c}\), the encryption of a random $\bar{x}$ of length $\vert \M[k]\vert$. We have

\[ \SS \equiv \H'_r \text{ and } \I \equiv \H'_0 \]

and

\[ \H'_i \equiv \G_0^\UPDCPA \CC'_i \equiv \G_1^\UPDCPA \CC'_{i+1} \]

Indeed, this can be seen on the following timeline (\ref{tab:orgcca2914}) 

\begin{table}[htbp]
\centering
\begin{tabular}{|c|c|c|c|c|}
\hline
\(j\)-th update query for $\M[k]$ & \(j<i\) & \(j=i\) & \(j=i+1\) & \(j>i+1\)\\
\hline
\({\color{green}\G_0^\UPDCPA} \CC'_i\) & \(\Upd(c)\) & \color{green}{$\Upd(c)$} & \(\Upd(\bar{c})\) & \(\Upd(\bar{c})\)\\
\hline
\({\color{green}\G_1^\UPDCPA} \CC'_{i+1}\) & \(\Upd(c)\) & \(\Upd(c)\) & \color{green}{$\Upd(\bar{c})$} & \(\Upd(\bar{c})\)\\
\hline
\end{tabular}
\caption{\label{tab:orgcca2914}Leakage behavior of both systems for each update request ($\bar{c}$ is always the encryption of a random plaintext of length $\vert \M[k]\vert$).}

\end{table}

Let \(\CC'_I\) be a reduction that samples \(i\in [1, r]\) at random and behaves like \(\CC'_i\) and define \(\D'' := \D\CC'_I\). We have,

\[ \Pr[\D''(\G_0^\UPDCPA) = 1] = \frac{1}{r} \sum_{i = 1}^r{\Pr[\D(\CC'_i \G_0^\UPDCPA) = 1]} \]

and

\begin{align*}
    \Pr[\D''(\G_1^\UPDCPA) = 1] & = \frac{1}{r} \sum_{i=1}^r{\Pr[\D(\CC'_i \G_1^\UPDCPA) = 1]} \\
     & = \frac{1}{r} \sum_{i=0}^{r-1}{\Pr[\D(\CC'_i \G_0^\UPDCPA) = 1]}
    \end{align*}

Finally, the advantage of the distinguisher in distinguishing system \(\SS\) from \(\I\) is

\begin{align*}
  \Delta^\D (\SS, \I) & = \Delta^\D(\H'_r, \H'_0) \\
                      & = \Delta^\D(\CC'_r \G_0^\UPDCPA, \CC'_0 \G_0^\UPDCPA) \\
                      & = \vert \Pr[\D(\CC'_r \G_0^\UPDCPA) = 1] - \Pr[\D(\CC'_0 \G_0^\UPDCPA) = 1] \vert \\
                      & = \vert \sum_{i=1}^r{\Pr[\D(\CC'_i \G_0^\UPDCPA) = 1]} - \sum_{i=0}^{r-1}{\Pr[\D(\CC'_i \G_0^\UPDCPA) = 1]} \vert \\
                      & = r \cdot \vert \Pr[\D''(\G_0^\UPDCPA) = 1] - \Pr[\D''(\G_1^\UPDCPA) = 1] \vert \\
                      & = r\cdot \Delta^{\D''}(\G_0^\UPDCPA, \G_1^\UPDCPA)
\end{align*}

\medskip

\noindent
\textbullet \(\:\) We use the triangular inequality to conclude. Let \(q\) be our upper bound on the number of writes and \(r\) be our upper bound on the number of updates. The advantage of the distinguisher in distinguishing the real system \(\R\) from the ideal one \(\I\) is

\begin{align*}
\Delta^\D(\R, \I) & \leq \Delta^\D(\R, \SS) + \Delta^\D(\SS, \I) \\
                  & = q\cdot \Delta^{\D'}(\G_0^\ENCCPA, \G_1^\ENCCPA) + r\cdot \Delta^{\D''}(\G_0^\UPDCPA, \G_1^\UPDCPA)
\end{align*}

So \(\ENCCPA+\UPDCPA\) is sufficient to securely construct the \(\cUSMR^+\) in the unbounded leakage model, where the age of each database entry is not hidden. \(\square\)
\end{proof}

%\end{subappendices}

\end{document}